\newif\iferrata

 \erratafalse

\documentclass{article}

\usepackage[final]{neurips_2023}

\usepackage[utf8]{inputenc} 
\usepackage[T1]{fontenc}    
\usepackage{hyperref}       
\usepackage{url}            
\usepackage{booktabs}       
\usepackage{amsfonts}       
\usepackage{nicefrac}       
\usepackage{microtype}      
\usepackage{xcolor}         
\usepackage{amsmath,amssymb,amsfonts,amsthm}
\usepackage{cleveref}
\usepackage{array}
\usepackage{enumerate}
\usepackage{thm-restate}
\usepackage{todonotes}
\usepackage{subfigure}
\usepackage{graphicx}

\newtheorem{theorem}{Theorem}
\newtheorem{lemma}[theorem]{Lemma}
\newtheorem{corollary}[theorem]{Corollary}

\newtheorem{observation}[theorem]{Observation}

\newcommand{\bbR}{\mathbb R} 

\newcommand{\calS}{\mathcal S}
\newcommand{\calT}{\mathcal T}

\newcommand{\calC}{\mathcal C}

\newcommand{\calO}{\mathcal O}

\newcommand{\calL}{\mathcal L}
\newcommand{\calQ}{\mathcal Q}

\newcommand{\ld}{\left}
\newcommand{\rd}{\right}
\newcommand{\norm}[1]{\left|\left|#1\right|\right|}

\newcommand{\abs}[1]{\left | #1 \right |}
\newcommand{\set}[1]{\left \{ #1 \right \}}
\newcommand{\cond}{\;|\;}

\newcommand{\eps}{\varepsilon}
\newcommand{\Prp}[1]{\Pr\!\left[{#1} \right]}

\newcommand{\alg}{\text{ALG}}
\newcommand{\opt}{\text{OPT}}

\newcommand{\optset}{\ensuremath{{\mathcal O}^*}}
\newcommand{\optcl}{\ensuremath{O^*}}

\newcommand{\tilo}{\widetilde{\calO}}
\newcommand{\tilc}{\widetilde{\calC}}

\newcommand{\blue}[1]{{\textcolor{blue}{{#1}} }}

\newcommand{\algokmeans}{\textsc{KM++}}
\newcommand{\algomslsarg}[1]{\textsc{MSLS-}\ensuremath{p = #1}}
\newcommand{\algogreedymslsarg}[1]{\textsc{MSLS-G-}\ensuremath{p = #1}}
\newcommand{\algomsls}{\textsc{MSLS}}
\newcommand{\algogreedymsls}{\textsc{MSLS-G}}
\newcommand{\algossls}{\textsc{SSLS}}
\newcommand{\approxcenters}{\textsc{APX-centers}}

\newcommand{\reassign}[2]{\texttt{Reassign}\!\ld(#1, #2\rd)}
\newcommand{\core}[1]{\texttt{core}\!\ld(#1\rd)}

\newcommand{\cost}[2]{\texttt{cost}\!\ld(#1, #2\rd)}
\newcommand{\ho}{\hat o}
\newcommand{\whin}{\widehat{In}}

\newcommand{\ignore}[1]{}

\title{Multi-Swap $k$-Means++}

\author{%
  Lorenzo Beretta\thanks{Authors are ordered in alphabetical order.} \\
  University of Copenhagen\\
  \texttt{lorenzo2beretta@gmail.com} \\
  \And
  Vincent Cohen-Addad \\
  Google Research\\
  \texttt{cohenaddad@google.com} \\
  \And 
  Silvio Lattanzi \\
  Google Research\\
  \texttt{silviol@google.com} \\
  \And
  Nikos Parotsidis \\
  Google Research\\
  \texttt{nikosp@google.com} \\
  \\
  \texttt{\url{https://github.com/lorenzo2beretta/multi-swap-k-means-pp}}
}

\begin{document}

\maketitle


\begin{abstract}
The $k$-means++ algorithm of Arthur and Vassilvitskii (SODA 2007) is often the practitioners' choice algorithm for optimizing the popular $k$-means clustering objective and is known to give an $O(\log k)$-approximation in expectation. To obtain higher quality solutions, Lattanzi and Sohler (ICML 2019) proposed augmenting $k$-means++ with $O(k \log \log k)$ local search steps obtained through the $k$-means++ sampling distribution to yield a $c$-approximation to the $k$-means clustering problem, where $c$ is a large absolute constant. Here we generalize and extend their local search algorithm by considering larger and more sophisticated local search neighborhoods hence allowing to  swap multiple centers at the same time. Our algorithm achieves a $9 + \varepsilon$ approximation ratio, which is the best possible for local search. Importantly we show that our approach yields substantial 
practical improvements, we show significant quality improvements over 
the approach of Lattanzi and Sohler (ICML 2019) on several datasets.
\end{abstract}

\section{Introduction}
\vspace{-0.2cm}
Clustering is a central problem in unsupervised learning. In clustering one is interested in grouping together ``similar'' object and separate ``dissimilar'' one. Thanks to its popularity many notions of clustering have been proposed overtime. In this paper, we focus on metric clustering and on one of the most studied problem in the area: the Euclidean $k$-means problem. 

In the Euclidean $k$-means problem one is given in input a set of points $P$ in $\mathbb{R}^d$. The goal of the problem is to find a set of $k$ centers so that the sum of the square distances to the centers is minimized. More formally, we are interested in finding a set $C$ of $k$ points in $\mathbb{R}^d$ such that $\sum_{p\in P}\min_{c\in C}\norm{p-c}^2$, where with $\norm{p - c}$ we denote the Euclidean distance between $p$ and $c$.

The $k$-means problem has a long history, in statistics
and operations research. For Euclidean $k$-means 
with running time polynomial in both $n, k$ and $d$, a $5.912$-approximation was recently shown 
in~\cite{Cohen-AddadEMN22}, improving upon \cite{kanungo,ahmadian2017better,Grandoni21} by leveraging the properties of the Euclidean metric. In terms of lower
bounds, the first to show that the high-dimensional 
$k$-means problems were APX-hard were \cite{GI03},
and later ~\cite{AwasthiCKS15} showed that the APX-hardness
holds even if the centers can be placed arbitrarily in $\mathbb{R}^d$. The inapproximability bound was later slightly improved by ~\cite{DBLP:journals/ipl/LeeSW17} until the recent best known 
bounds of~\cite{Cohen-AddadS19,Cohen-AddadLS22,DBLP:conf/soda/Cohen-AddadSL21}
that showed that it is NP-hard to achieve a better than
1.06-approximation and hard to approximate it better than 1.36 
assuming a stronger conjecture.
From a more practical point of view, \cite{ArV09} showed that the widely-used popular heuristic of Lloyd~\cite{lloyd1957least} can lead to solutions with arbitrarily
bad approximation guarantees, but can be improved by a simple seeding strategy, called $k$-means++, so as to guarantee that the output is within an $O(\log k)$ factor of the optimum~\cite{kmpp-original}.

Thanks to its simplicity $k$-means++ is widely adopted in practice. In an effort to improve its performances \cite{silvio-original, few-more-steps} combine $k$-means++ and local search to efficiently obtain a constant approximation algorithm with good practical performance. These two studies show that one can use the $k$-means++
distribution in combination with a local search algorithm to get the 
best of both worlds: a practical algorithm with constant approximation
guarantees.

However, the constant obtained in  \cite{silvio-original, few-more-steps} 
is very large (several thousands in theory) and the question as whether 
one could obtain  a practical algorithm that would efficiently match the 
$9+\eps$-approximation obtained by the $n^{O(d/\epsilon)}$ algorithm of \cite{kanungo} has remained open. Bridging the gap between the 
theoretical approach of \cite{kanungo} and $k$-means++ has thus been
a long standing goal.

\vspace{-0.2cm}
\paragraph{Our Contributions.}
We make significant progress on the above line of work.
\vspace{-0.2cm}
\begin{itemize}
    \item We adapt techniques from the analysis of \cite{kanungo} to obtain a tighter analysis of the algorithm in \cite{silvio-original}. In particular in \Cref{cor:1-swap-better-analysis}, we show that their algorithm achieves an approximation of ratio of $\approx 26.64$.
    \item We extend this approach to multi-swaps, where we allow swapping more than one center at each iteration of local search, improving significantly the approximation to $\approx 10.48$ in time $O(nd \cdot poly(k))$. 
    \item Leveraging ideas from \cite{power-means}, we design a better local search swap that improves the approximation further to $9+\eps$ (see \Cref{thm:main-result}). This new algorithm matches the $9+\eps$-approximation achieved by the local search algorithm in~\cite{kanungo}, but it is significantly more efficient. Notice that $9$ is the best approximation achievable through local search algorithms, as proved in ~\cite{kanungo}.   
    \item We provide experiments where we compare against $k$-means++ and \cite{silvio-original}. We study a variant of our algorithm that performs very competitively with our theoretically sound algorithm. The variant is very efficient and still outperforms previous work in terms of solution quality, even after the standard postprocessing using Lloyd.
\end{itemize}

\vspace{-0.2cm}
\paragraph{Additional Related Work.}
We start by reviewing the approach of \cite{kanungo} and a possible adaptation to our setting. The bound of $9+\eps$
on the approximation guarantee shown by \cite{kanungo} is for the 
following algorithm: Given a set $S$ of $k$ centers, if there is a set 
$S^+$ of at most $2/\eps$ points  in $\mathbb{R}^d$ together with a set $S^-$ of
$|S^+|$ points in $S$ such that $S \setminus S^- \cup S^+$ achieves a better $k$-means cost than $S$, then set $S := S \setminus S^- \cup S^+$ and repeat until convergence.
The main drawback of the algorithm is that it asks whether there
exists a set $S^+$ of points in $\mathbb{R}^d$ that could be swapped
with elements of $S$ to improve the cost. Identifying such a set, even
of constant size, is already non-trivial. The best way of doing so is through the following path: First compute a coreset using the 
state-of-the-art coreset construction of~\cite{DBLP:conf/stoc/Cohen-AddadLSS22} and apply the 
dimensionality reduction
of \cite{BecchettiBC0S19,MakarychevMR19}, hence obtaining a set of 
$\tilde{O}(k/\eps^{4})$ points in 
dimension $O(\log k /\eps^{2})$. Then, compute grids using the 
discretization framework of \cite{Mat00} to identify a set of 
$\eps^{-O(d)} \sim
k^{O(\eps^{-2} \log (1/\eps))}$ grid points that contains nearly-optimum 
centers. Now, run the local search algorithm where the sets $S^+$ are
chosen from the grid points by brute-force enumeration over all possible
subsets of grid points of size at most, say $s$.
The running time of the whole algorithm with swaps of magnitude $s$,  i.e.: $|S^+| \le s$, 
hence becomes  $k^{O(s\cdot \eps^{-2} \log (1/\eps))}$  
for an approximation of $(1+\eps)(9+2/s)$, meaning a
dependency in $k$ of $k^{O(\eps^{-3} \log (1/\eps))}$ 
to achieve a $9+\eps$-approximation.
Our results improves upon this approach in two ways: (1) it improves over the above theoretical bound 
and (2) does so through an efficient and implementable, i.e.: practical, algorithm. 

Recently, \cite{DBLP:conf/soda/GrunauORT23} looked at how much applying a greedy rule on top of the $k$-means++ heuristic improves its performance. The heuristic
is that at each step, the algorithm 
samples $\ell$ centers and only keeps
the one that gives the best improvement in cost. Interestingly the authors prove that from a theoretical standpoint this heuristic does not improve the quality of the output. Local search algorithms for $k$-median
and $k$-means have also been studied by \cite{GuT08} who drastically 
simplified the analysis of \cite{AryaGKMMP04}. \cite{Cohen-AddadS17} 
demonstrated the power of local search for stable instances. \cite{FriggstadRS19,Cohen-AddadKM19} showed that local search yields a 
PTAS for Euclidean inputs of bounded dimension (and doubling metrics)
and minor-free metrics. \cite{Cohen-Addad18} showed how to speed up the local search algorithm using $kd$-trees (i.e.: for low dimensional inputs).

For fixed $k$, there are several known approximation schemes, typically using small coresets \cite{BecchettiBC0S19,FL11,KumarSS10}.  The state-of-the-art
approaches are due to~\cite{BhattacharyaGJ020,Jaiswal0S14}.
The best known coreset construction remains~\cite{DBLP:conf/nips/Cohen-AddadLSSS22,DBLP:conf/stoc/Cohen-AddadLSS22}.

If the constraint on the number of output centers is relaxed, then
we talk about bicriteria approximations and $k$-means
has been largely studied~\cite{BaV15,CharikarG05,CoM15,KPR00,MakarychevMSW16}.

\vspace{-0.25cm}
\section{Preliminaries}
\vspace{-0.2cm}
\label{sec:preliminaries}
\paragraph{Notation.}
We denote with $P \subseteq \bbR^d$ the set of input points and let $n = |P|$.
Given a point set $Q\subseteq P$ we use $\mu(Q)$ to denote the mean of points in $Q$.
Given a point $p \in P$ and a set of centers $A$ we denote with $A[p]$ the closest center in $A$ to $p$ (ties are broken arbitrarily).
We denote with $\calC$ the set of centers currently found by our algorithm and with $\optset$ an optimal set of centers.
Therefore, given $p \in P$, we denote with $\calC[p]$ and $\optset[p]$ its closest $\alg$-center and $\opt$-center respectively.
We denote by $\cost{Q}{A}$ the cost of points in $Q \subseteq P$ w.r.t. the centers in $A$, namely
\[
\cost{Q}{A} = \sum_{q \in Q} \min_{c \in A} \norm{q - c}^2.
\]
We use $\alg$ and $\opt$ as a shorthand for $\cost{P}{\calC}$ and $\cost{P}{\optset}$ respectively. When we sample points proportionally to their current cost (namely, sample $q$ with probability $\cost{q}{\calC} / \cost{P}{\calC}$) we call this the $D^2$ distribution. 
When using $O_\varepsilon(\cdot)$ and $\Omega_\varepsilon(\cdot)$ we mean that $\varepsilon$ is considered constant.  
We use $\widetilde O(f)$ to hide polylogarithmic factors in $f$. The following lemma is folklore.

\begin{lemma} \label{lem:moved-cost}
Given a point set $Q \subseteq P$ and a point $p\in P$ we have
\[
\cost{Q}{p} = \cost{Q}{\mu(Q)} + \abs{Q} \cdot \norm{p - \mu(Q)}^2.
\]
\end{lemma}

Let $\optcl_i$ be an optimal cluster, we define the \emph{radius} of $\optcl_i$ as $\rho_i$ such that $\rho_i^2 \cdot |\optcl_i| = \cost{\optcl_i}{o_i}$, where $o_i = \mu(\optcl_i)$. 
We define the $\delta$-\emph{core} of the optimal cluster $\optcl_i$ as the set of points $p\in \optcl_i$ that lie in a ball of radius $(1+\delta)\rho_i$ centered in $o_i$. In symbols, $\core{\optcl_i} = P \cap B(o_i, (1+\delta)\rho_i)$. Throughout the paper, $\delta$ is always a small constant fixed upfront, hence we omit it.

\begin{lemma} \label{lem:sample-from-core}
Let $\optcl_i$ be an optimal cluster and sample $q \in \optcl_i$ according to the $D^2$-distribution restricted to $\optcl_i$. 
If $\cost{\optcl_i}{\calC} > (2 + 3\delta) \cdot \cost{\optcl_i}{o_i}$ then $\Prp{q \in \core{\optcl_i}} = \Omega_\delta(1)$.  
\end{lemma}
\vspace{-0.2cm}
\begin{proof}
Define $\alpha := \cost{\optcl_i}{\calC}/ \cost{\optcl_i}{o_i} > 2 + 3\delta$. Thanks to \Cref{lem:moved-cost}, for each $c \in \calC$  we have $\norm{c - o_i}^2 \geq (\alpha - 1) \rho_i^2$.  Therefore, for each $y \in \core{\optcl_i}$ and every $c \in \calC$ we have 
\[
\cost{y}{c} = ||y-c||^2 \geq \ld(\sqrt{\alpha - 1} - (1 + \delta)\rd)^2 \cdot \rho_i^2 = \Omega_\delta(\alpha \rho_i^2). 
\]
Moreover, by a Markov's inequality argument we have $\abs{\optcl_i \setminus \core{\optcl_i}} \leq \frac{1}{1 + \delta} \cdot \abs{\optcl_i}$
and thus $ \abs{\core{\optcl_i}} \geq \Omega_\delta(|\optcl_i|)$.
Combining everything we get
\[
\cost{\core{\optcl_i}}{\calC} \geq |\core{\optcl_i}| \cdot \min_{\substack{c \in \calC \\ y \in \core{\optcl_i}}} \cost{y}{c} = \Omega_\delta(|\optcl_i|) \cdot \Omega_\delta (\alpha \rho_i^2)
\]
and $|\optcl_i| \cdot \alpha \rho_i^2 = \cost{\optcl_i}{\calC}$, hence $\cost{\core{\optcl_i}}{\calC} = \Omega_\delta(\cost{\optcl_i}{\calC})$.
\end{proof}

\vspace{-0.25cm}
\section{Multi-Swap $k$-Means++} 
\vspace{-0.2cm}

\label{sec:multi-swap-kmpp}
The single-swap local search (SSLS) $k$-means++ algorithm in \cite{silvio-original} works as follows. First, $k$ centers are sampled using $k$-means++ (namely, they are sampled one by one according to the $D^2$ distribution, updated for every new center). Then, $O(k \log \log k)$ steps of local search follow. In each local search step a point $q \in P$ is $D^2$-sampled, then let $c$ be the center among the current centers $\calC$ such that $\cost{P}{(\calC \setminus \{c\}) \cup \{q\}}$ is minimum. If $\cost{P}{(\calC \setminus \{c\}) \cup \{q\}} < \cost{P}{\calC}$ then we swap $c$ and $q$, or more formally we set $\calC \leftarrow (\calC \setminus \{c\}) \cup \{q\}$. 

We extend the SSLS so that we allow to swap multiple centers simultaneously and call this algorithm multi-swap local search (MSLS) $k$-means++. Swapping multiple centers at the same time achieves a lower approximation ratio, in exchange for a higher time complexity. 
In this section, we present and analyse the $p$-swap local search (LS) algorithm for a generic number of $p$ centers swapped at each step. For any constant $\delta >0$, we obtain an approximation ratio $\alg / \opt = \eta^2 + \delta$ where 
\begin{equation} \label{eq:eta-multi-swap}
\eta^2 - (2+2/p) \eta - (4 + 2/p) = 0.
\end{equation} 

\vspace{-0.3cm}
\paragraph{The Algorithm.}
First, we initialize our set of centers using $k$-means++. Then, we run $O(ndk^{p-1})$ local search steps, where a local search step works as follows. We $D^2$-sample a set $In = \set{q_1 \dots q_p}$ of points from $P$ (without updating costs). Then, we iterate over all possible sets $Out = \set{c_1 \dots c_p}$ of $p$ distinct elements in $\calC \cup In$ and select the set $Out$ such that performing the swap $(In, Out)$ maximally improves the cost\footnote{If $In \cap Out \neq \emptyset$ then we are actually performing the swap $(In \setminus Out, Out \setminus In)$ of size $< p$.}. 
If this choice of $Out$ improves the cost, then we perform the swap $(In, Out)$, else we do not perform any swap for this step.

\begin{theorem} \label{thm:multi-swap-analysis}
For any $\delta > 0$, the $p$-swap local search algorithm above runs in $\widetilde O(ndk^{2p})$ time and, with constant probability, finds an $(\eta^2 + \delta)$-approximation of $k$-means, where $\eta$ satisfies \Cref{eq:eta-multi-swap}.
\end{theorem}

Notice that the SSLS algorithm of \cite{silvio-original} is exactly the $p$-swap LS algorithm above for $p=1$.
\begin{corollary} \label{cor:1-swap-better-analysis}
The single-swap local search in \cite{silvio-original, few-more-steps} achieves an approximation ratio $< 26.64$.
\end{corollary}

\begin{corollary} \label{cor:msls-10.64-apx}
For $p = O(1)$ large enough, multi-swap local search achieves an approximation ratio $< 10.48$ in time $O(nd\cdot  poly(k))$.
\end{corollary}

\vspace{-0.2cm}
\subsection{Analysis of Multi-Swap $k$-means++}
\vspace{-0.1cm}
In this section we prove \Cref{thm:multi-swap-analysis}. Our main stepping stone is the following lemma.

\begin{lemma} \label{lem:multi-swap-improvement}
Let $\alg$ denote the cost at some point in the execution of MSLS. 
As long as $\alg/\opt > \eta^2 + \delta$, a local search step improves the cost by a factor $1 - \Omega(1/k)$ with probability $\Omega(1/k^{p-1})$. 
\end{lemma}
\begin{proof}[Proof of \Cref{thm:multi-swap-analysis}]
First, we show that $O(k^{p} \log\log k)$ local steps suffice to obtain the desired approximation ratio, with constant probability.
Notice that a local search step can only improve the cost function, so it is sufficient to show that the approximation ratio is achieved at some point in time.
We initialize our centers using $k$-means++, which gives a $O(\log k)$-approximation in expectation. Thus, using Markov's inequality the approximation guarantee $O(\log k)$ holds with arbitrary high constant probability. We say that a local-search step is \emph{successful} if it improves the cost by a factor of at least $1-\Omega(1/k)$. 
Thanks to \Cref{lem:multi-swap-improvement}, we know that unless the algorithm has already achieved the desired approximation ratio then a local-search step is successful with probability $\Omega(1/k^{p-1})$.
To go from $O(\log k)$ to $\eta^2 + \delta$ we need $O(k \log \log k)$ successful local search steps.
Standard concentration bounds on the value of a Negative Binomial random variable show that, with high probability, the number of trial to obtain $O(k \log \log k)$ successful local-search steps is $O(k^p \log \log k)$. 
Therefore, after $O(k^p \log \log k)$ local-search steps we obtain an approximation ratio of $\eta^2 + \delta$. 

To prove the running time bound it is sufficient to show that a local search step can be performed in time $\widetilde O(ndk^{p-1})$. This is possible if we maintain, for each point $x \in P$, a dynamic sorted dictionary\footnote{Also known as dynamic predecessor search data structure.} storing the pairs $(\cost{x}{c_i}, c_i)$ for each $c_i \in \calC$. Then we can combine the exhaustive search over all possible size-$p$ subsets of $\calC \cup In$ and the computation of the new cost function using time $O(ndk^{p-1} \log k)$. To do so, we iterate over all possible size-$(p-1)$ subsets $Z$ of $\calC \cup In$ and update all costs as if these centers were removed, then for each point $x \in P$ we compute how much its cost increases if we remove its closest center $c_x$ in $(\calC \cup In) \setminus Z$ and charge that amount to $c_x$. In the end, we consider $Out = Z \cup \{c\}$ where $c$ is the cheapest-to-remove center found in this way.
\end{proof}

\vspace{-0.1cm}
The rest of this section is devoted to proving \Cref{lem:multi-swap-improvement}. 
For convenience, we prove that \Cref{lem:multi-swap-improvement} holds whenever $\alg / \opt > \eta^2 + O(\delta)$, which is wlog by rescaling $\delta$. Recall that we now focus on a given step of the algorithm, and when we say current cost, current centers and current clusters we refer to the state of these objects at the end of the last local-search step before the current one. 
Let $\optcl_1 \dots \optcl_k$ be an optimal clustering of $P$ and let $\optset = \set{o_i = \mu(\optcl_i) \cond \text{ for } i = 1 \dots k}$ be the set of optimal centers of these clusters. We denote with $C_1 \dots C_k$ the current set of clusters at that stage of the local search and with $\calC = \{c_1 \dots c_k\}$ the set of their respective current centers.

We say that $c_i$ \emph{captures} $o_j$ if $c_i$ is the closest current center to $o_j$, namely $c_i = \calC[o_j]$. We say that $c_i$ is \emph{busy} if it captures more than $p$ optimal centers, and we say it is \emph{lonely} if it captures no optimal center. 
Let $\tilo = \set{o_i \cond \cost{\optcl_i}{\calC} > \delta \cdot \alg / k}$ and $\tilc = \calC \setminus \{\calC[o_i] \cond o_i \in \optset \setminus \tilo\}$. For ease of notation, we simply assume that $\tilo = \{o_1 \dots o_h\}$ and $\tilc = \{c_1 \dots c_{h'}\}$. Notice that $h' > h$.

\vspace{-0.2cm}
\paragraph{Weighted ideal multi-swaps.} 
Given $In \subseteq P$ and $Out \subseteq \tilc$ of the same size we say that the swap $(In, Out)$ is an \emph{ideal} swap if $In \subseteq \tilo$.
We now build a set of \emph{weighted} ideal multi-swaps $\calS$. First, suppose wlog that $\set{c_1 \dots c_t}$ is the set of current centers in $\tilc$ that are neither lonely nor busy.
Let $\calL$ be the set of lonely centers in $\tilc$. For each $i=1 \dots t$, we do the following. 
Let $In$ be the set of optimal centers in $\tilo$ captured by $c_i$. 
Choose a set $\calL_i$ of $|In|-1$ centers from $\calL$, set $\calL \leftarrow \calL \setminus \calL_i$ and define $Out = \calL_i \cup \{c_i\}$. Assign weight $1$ to $(In, Out)$ and add it to $\calS$.
For each busy center $c_i\in \set{c_{t+1} \dots c_{h'}}$ let $A$ be the set of optimal centers in $\tilo$ captured by $c_i$, pick a set $\calL_i$ of $|A| - 1$ lonely current centers from $\calL$ (a counting argument shows that this is always possible). Set $\calL \leftarrow \calL \setminus \calL_i$. For each $o_j \in A$ and $c_\ell \in \calL_i$ assign weight $1/(|A| - 1)$ to $(o_j, c_\ell)$ and add it to $\calS$.
Suppose we are left with $\ell$ centers $o'_1 \dots o'_\ell \in \tilo$ such that $\calC[o'_i] \not\in \tilc$. Apparently, we have not included any $o'_i$ in any swap yet. However, since $|\tilc| \geq |\tilo|$, we are left with at least $\ell' \geq \ell$ lonely centers $c'_1 \dots c'_{\ell'} \in \tilc$. For each $i =1 \dots \ell$ we assign weight $1$ to $(o'_i, c'_i)$ and add it to $\calS$.

\begin{observation} \label{obs:ideal-swaps-weights}
The process above generates a set of weighted ideal multi-swaps such that: (i) Every swap has size at most $p$; (ii) The combined weights of swaps involving an optimal center $o_i \in \tilo$ is $1$; (iii) The combined weights of swaps involving a current center $c_i$ is at most $1 + 1/p$.   
\end{observation}

Consider an ideal swap $(In, Out)$. Let $\optcl_{In} = \bigcup_{o_i \in In} \optcl_i$ and $C_{Out} = \bigcup_{c_j \in Out} C_j$. Define the reassignment cost $\reassign{In}{Out}$ as the increase in cost of reassigning points in $C_{Out} \setminus \optcl_{In}$ to centers in $\calC \setminus Out$. Namely, 
\[
\reassign{In}{Out} =  \cost{C_{Out} \setminus \optcl_{In}}{\calC \setminus Out} - \cost{C_{Out} \setminus \optcl_{In}}{\calC}.
\]
We take the increase in cost of the following reassignment as an upper bound to the reassignment cost. For each $p \in C_{Out} \setminus \optcl_{In}$ we consider its closest optimal center $\optset[p]$ and reassign $p$ to the current center that is closest to $\optset[p]$, namely $\calC[\optset[p]]$. In formulas, we have
\begin{align*}
\reassign{In}{Out} &\leq \sum_{p \in C_{Out} \setminus \optcl_{In}} \cost{p}{\calC[\optset[p]]} - \cost{p}{\calC[p]} \\ &\leq \sum_{p \in C_{Out}} \cost{p}{\calC[\optset[p]]} - \cost{p}{\calC[p]}.
\end{align*}
Indeed, by the way we defined our ideal swaps we have $\calC[\optset[p]] \not \in Out$ for each $p \not\in \optcl_{In}$ and this reassignment is valid. Notice that the right hand side in the equation above does not depend on $In$. 

\begin{restatable}{lemma}{CombinedReassignmentCost}
\label{lem:reassignment-cost-tech}
$\sum_{p \in P} \cost{p}{\calC[\optset[p]]} \leq 2\opt + \alg + 2\sqrt{\alg} \sqrt{\opt}$.
\end{restatable}
\vspace{-0.3cm}
\begin{proof} Deferred to the supplementary material.
\end{proof}

\begin{lemma} \label{lem:combined-reassign-cost-ms}
The combined weighted reassignment costs of all ideal multi-swaps in $\calS$ is at most $(2+2/p) \cdot (\opt + \sqrt{\alg}\sqrt{\opt})$. 
\end{lemma}
\vspace{-0.3cm}
\begin{proof}
Denote by $w(In, Out)$ the weight associated with the swap $(In, Out)$.
\begin{align*}
\sum_{(In, Out) \in \calS} w(In, Out) \cdot \reassign{In}{Out} &\leq \\
\sum_{(In, Out) \in \calS} w(In, Out) \cdot \sum_{p \in C_{Out}} \cost{p}{\calC[\optset[p]]} - \cost{p}{\calC[p]} &\leq \\
(1+1/p) \cdot \sum_{c_j \in \calC} \sum_{p \in C_j} \cost{p}{\calC[\optset[p]]} - \cost{p}{\calC[p]} &\leq \\
(1+1/p) \cdot \left(\sum_{p \in P} \cost{p}{\calC[\optset[p]]} - \alg\right).
\end{align*}
The second inequality uses $(iii)$ from \Cref{obs:ideal-swaps-weights}. Applying \Cref{lem:reassignment-cost-tech} completes the proof.
\end{proof}
\vspace{-0.1cm}

Recall the notions of radius and core of an optimal cluster introduced in \Cref{sec:preliminaries}.
We say that a swap $(In, Out)$ is \emph{strongly improving} if $\cost{P}{(\calC \cup In) \setminus Out} \leq (1 - \delta/ k ) \cdot \cost{P}{\calC}$.
Let $In = \{o_1 \dots o_s\} \subseteq \tilo$ and $Out = \{c_1 \dots c_s \} \subseteq \tilc$ we say that an ideal swap $(In, Out)$ is \emph{good} if for every $q_1 \in \core{o_1} \dots q_s \in \core{o_s}$ the swap $(\calQ, Out)$ is strongly improving, where $\calQ = \set{q_1 \dots q_s}$. We call an ideal swap \emph{bad} otherwise. 
We say that an optimal center $o_i \in \tilo$ is good if that's the case for at least one of the ideal swaps it belongs to, otherwise we say that it is bad.
Notice that each optimal center in $\tilo$ is assigned to a swap in $\calS$, so it is either good or bad.
Denote with $G$ the union of cores of good optimal centers in $\tilo$. 

\begin{lemma} \label{lem:bad-swap-inequality}
If an ideal swap $(In, Out)$ is bad, then we have
\begin{equation} \label{eq:bad-swap}
\cost{\optcl_{In}}{\calC} \leq (2+\delta) \cost{\optcl_{In}}{\optset} + \reassign{In}{Out} + \delta \alg / k.
\end{equation}
\end{lemma}
\vspace{-0.3cm}
\begin{proof}
Let $In = \{o_1 \dots o_s\}$, $\calQ = \set{q_1 \dots q_s}$ such that $q_1 \in \core{o_1} \dots q_s \in \core{o_s}$. Then, by \Cref{lem:moved-cost} $\cost{\optcl_{In}}{\calQ} \leq (2+\delta) \cost{\optcl_{In}}{\optset}$. Moreover, $\reassign{In}{Out} = \cost{P \setminus \optcl_{In}}{\calC \setminus Out} - \cost{P \setminus \optcl_{In}}{\calC}$ because points in $P \setminus C_{Out}$ are not affected by the swap. 
Therefore, $\cost{P}{(\calC \cup \calQ) \setminus Out} \leq (2+\delta) \cost{\optcl_{In}}{\optset} + \reassign{In}{Out} + \cost{P \setminus \optcl_{In}}{\calC}$. Suppose by contradiction that \Cref{eq:bad-swap} does not hold, then
\vspace{-0.1cm}
\begin{align*}
\cost{P}{\calC} - \cost{P}{(\calC \cup \calQ) \setminus Out} &= \\ \cost{P \setminus \optcl_{In}}{\calC} + \cost{\optcl_{In}}{\calC} - \cost{P}{(\calC \cup \calQ) \setminus Out} &\geq \delta \alg / k.
\end{align*}
Hence, $(\calQ, Out)$ is strongly improving and this holds for any choice of $\calQ$, contradiction.
\end{proof}

\begin{lemma} \label{lem:all-qs-in-cores}
If $\alg  /\opt > \eta^2 + \delta$ then $\cost{G}{\calC} = \Omega_\delta(\cost{P}{\calC})$. Thus, if we $D^2$-sample $q$ we have $P[q \in G] = \Omega_\delta(1)$.
\end{lemma}
\vspace{-0.3cm}
\begin{proof}
First, we observe that the combined current cost of all optimal clusters in $\optset \setminus \tilo$ is at most $k \cdot \delta \alg / k = \delta \alg$. 
Now, we prove that the combined current cost of all $\optcl_i$ such that $o_i$ is bad is $\leq (1-2\delta)\alg$. Suppose, by contradiction, that it is not the case, then we have:
\begin{align*} \label{eq:combined-opt-reassignment}
(1-2\delta)\alg < \sum_{\text{ Bad } o_i \in \tilo} \cost{\optcl_i}{\calC} \leq 
\sum_{\text{ Bad } (In, Out) \in \calS} w(In, Out) \cdot \cost{\optcl_{In}}{\calC} &\leq \\
\sum_{\text{ Bad } (In, Out)} w(In, Out) \cdot \ld( (2+\delta) \cost{\optcl_{In}}{\optset} + \reassign{In}{Out} + \delta \alg / k \rd) &\leq \\
(2+\delta) \opt + (2 + 2/p)\opt + (2 + 2/p)\sqrt{\alg}\sqrt{\opt} + \delta \alg.
\end{align*}

The second and last inequalities make use of \Cref{obs:ideal-swaps-weights}. The third inequality uses \Cref{lem:bad-swap-inequality}.

Setting $\eta^2 = \alg / \opt$ we obtain the inequality $\eta^2 -(2 + 2/p \pm O(\delta)) \eta - (4 + 2/p \pm O(\delta)) \leq 0$.
Hence, we obtain a contradiction in the previous argument as long as $\eta^2 -(2 + 2/p \pm O(\delta)) \eta - (4 + 2/p \pm O(\delta)) > 0$. A contradiction there implies that at least an $\delta$-fraction of the current cost is due to points in $\bigcup_{\text{Good } o_i \in \tilo} \optcl_i$. We combine this with \Cref{lem:sample-from-core} and conclude that the total current cost of $G = \bigcup_{\text{Good } o_i \in \tilo} \core{\optcl_i}$ is $\Omega_\delta(\cost{P}{\calC})$.  
\end{proof}

\vspace{-0.3cm}
Finally, we prove \Cref{lem:multi-swap-improvement}. Whenever $q_1 \in G$ we have that $q_1 \in \core{o_1}$ for some good $o_1$. Then, for some $s \leq p$ we can complete $o_1$ with $o_2 \dots o_s$ such that $In  = \set{o_1 \dots o_s}$ belongs to a good swap. Concretely, there exists $Out \subseteq \calC$ such that $(In, Out)$ is a good swap. Since $In \subset \tilo$ we have $\cost{\optcl_i}{\calC} > \delta \opt / k$ for all $o_i \in In$, which combined with \Cref{lem:sample-from-core} gives that for $i = 2 \dots s$ $P[q_i \in \core{o_i}] \geq \Omega_\delta(1 / k)$. Hence, we have $P[q_i \in \core{o_i} \text{ for } i = 1 \dots s] \geq \Omega_{\delta, p}(1 / k^{p-1})$.
Whenever we sample $q_1 \dots q_s$ from $\core{o_1} \dots \core{o_s}$, we have that $(\calQ, Out)$ is strongly improving. Notice, however, that $(\calQ, Out)$ is a $s$-swap and we may have $s < p$. Nevertheless, whenever we sample $q_1 \dots q_s$ followed by any sequence $q_{s+1} \dots q_p$ it is enough to choose $Out' = Out \cup \{q_{s+1} \dots q_p\}$ to obtain that $(\{q_1 \dots q_p\}, Out')$ is an improving $p$-swap.

\vspace{-0.2cm}
\section{A Faster $(9+\varepsilon)$-Approximation Local Search Algorithm}
\vspace{-0.2cm} \label{sec:faster-9+eps-algo}

The MSLS algorithm from \Cref{sec:multi-swap-kmpp} achieves an approximation ratio of $\eta^2 + \varepsilon$, where $\eta^2 - (2 + 2/p) \eta - (\blue{4} + 2/p) = 0$ and $\varepsilon >0$ is an arbitrary small constant. For large $p$ we have $\eta \approx 10.48$. On the other hand, employing $p$ simultaneous swaps, \cite{kanungo} achieve an approximation factor of $\xi^2 + \varepsilon$ where $\xi^2 -(2 + 2/p) \xi - (\blue{3} + 2/p) = 0$. If we set $p \approx 1/\varepsilon$ this yields a $(9 + O(\varepsilon))$-approximation. In the same paper, they prove that $9$-approximation is indeed the best possible for $p$-swap local search, if $p$ is constant (see Theorem $3.1$ in \cite{kanungo}). 
They showed that $9$ is the right locality gap for local search, but they matched it with a very slow algorithm. To achieve a $(9+\varepsilon)$-approximation, they discretize the space reducing to $O(n \varepsilon^{-d})$ candidate centers and perform an exhaustive search over all size-$(1/\varepsilon)$ subsets of candidates at every step. As we saw in the related work section, it is possible to combine techniques from coreset and dimensionality reduction to reduce the number of points to $n' = k \cdot poly(\varepsilon^{-1})$ and the number of dimensions to $d' = \log k \cdot \varepsilon^{-2}$. This reduces the complexity of \cite{kanungo} to $k^{O(\varepsilon^{-3} \log \varepsilon^{-1})}$.

In this section, we leverage techniques from \cite{power-means} to achieve a $(9+\varepsilon)$-approximation faster
\footnote{The complexity in \Cref{thm:main-result} can be improved by applying the same preprocessing techniques using coresets and dimensionality reduction, similar to what can be used to speed up the approach of \cite{kanungo}. Our algorithm hence becomes asymptotically faster.}. In particular, we obtain the following.

\begin{restatable}{theorem}{MainTheorem} \label{thm:main-result}
Given a set of $n$ points in $\bbR^d$ with aspect ratio $\Delta$, there exists an algorithm that computes a $9+\varepsilon$-approximation to $k$-means in time $ndk^{O(\varepsilon^{-2})} \log^{O(\varepsilon^{-1})} (\Delta) \cdot 2^{-poly(\varepsilon^{-1})}$.
\end{restatable}

Notice that, besides being asymptotically slower, the pipeline obtained combining known techniques is highly impractical and thus it did not make for an experimental test-bed. Moreover, it is not obvious how to simplify such an ensemble of complex techniques to obtain a practical algorithm.

\vspace{-0.2cm}
\paragraph{Limitations of MSLS.} The barrier we need to overcome in order to match the bound in \cite{kanungo} is that, while we only consider points in $P$ as candidate centers, the discretization they employ considers also points in $\bbR^d \setminus P$. 
In the analysis of MSLS we show that we sample each point $q_i$ from $\core{\optcl_i}$ or equivalently that $q_i \in B(o_i, (1+\epsilon) \rho_i)$, where $\rho_i$ is such that $\optcl_i$ would have the same cost w.r.t. $o_i$ if all its points were moved on a sphere of radius $\rho_i$ centered in $o_i$. 
This allows us to use a Markov's inequality kind of argument and conclude that there must be $\Omega_\epsilon(|\optcl_i|)$ points in $\optcl_i \cap B(o_i, (1+\epsilon) \rho_i)$. 
However, we have no guarantee that there is any point at all in $\optcl_i \cap B(o_i, (1-\varepsilon) \rho_i)$. Indeed, all points in $\optcl_i$ might lie on $\partial B(o_i, \rho_i)$. The fact that potentially all our candidate centers $q$ are at distance at least $\rho_i$ from $o_i$ yields (by \Cref{lem:moved-cost}) $\cost{\optcl_i}{q} \geq 2 \cost{\optcl_i}{o_i}$, which causes the zero-degree term in $\xi^2 -(2 + 2/p) \xi - (\blue{3} + 2/p) = 0$ from \cite{kanungo} to become a $\blue{4}$ in our analysis.

\vspace{-0.2cm}
\paragraph{Improving MSLS by taking averages.} 
First, we notice that, in order to achieve $(9+\varepsilon)$-approximation we need to set $p = \Theta(1/\varepsilon)$. 
The main hurdle to achieve a $(9+\varepsilon)$-approximation is that we need to replace the $q_i$ in MSLS with a better approximation of $o_i$. We design a subroutine that computes, with constant probability, an $\varepsilon$-approximation $\hat o_i$ of $o_i$ (namely, $\cost{\optcl_i}{\ho_i} \leq (1+\varepsilon) \cost{\optcl_i}{o_i}$). The key idea is that, if sample uniformly $O(1/\varepsilon)$ points from $\optcl_i$ and define $\ho_i$ to be the average of our samples then $\cost{\optcl_i}{\ho_i} \leq (1+\varepsilon) \cost{\optcl_i}{o_i}$


Though, we do not know $\optcl_i$, so sampling uniformly from it is non-trivial. To achieve that, for each $q_i$ we identify a set $N$ of \emph{nice} candidate points in $P$ such that a $poly(\varepsilon) / k$ fraction of them are from $\optcl_i$. We sample $O(1/\varepsilon)$ points uniformly from $N$ and thus with probability $(\varepsilon / k)^{O(1/\varepsilon)}$ we sample only points from $\optcl_i$. Thus far, we sampled $O(1/\varepsilon)$ points uniformly from $N \cap \optcl_i$. What about the points in $\optcl_i \setminus N$? We can define $N$ so that all points in $\optcl_i \setminus N$ are either very close to some of the $(q_j)_j$ or they are very far from $q_i$. The points that are very close to points $(q_j)_j$ are easy to treat. Indeed, we can approximately locate them and we just need to guess their mass, which is matters only when $\geq poly(\varepsilon) \alg$, and so we pay only a $\log^{O(1/\varepsilon)}(1/\varepsilon)$ multiplicative overhead to guess the mass close to $q_j$ for $j = 1 \dots p = \Theta(1/\varepsilon)$. As for a point $f$ that is very far from $q_i$ (say, $||f-q_i|| \gg \rho_i$) we notice that, although $f$'s contribution to $\cost{\optcl_i}{o_i}$ may be large, we have $\cost{f}{o} \approx \cost{f}{o_i}$ for each $o \in B(q_i, \rho_i) \subseteq B(o_i, (2+\varepsilon) \rho_i)$ assuming $q_i \in \core{o_i}$.  

\vspace{-0.25cm}
\section{Experiments}
\vspace{-0.2cm}
In this section, we show that our new algorithm using multi-swap local search can be employed to design an efficient seeding algorithm for Lloyd's which outperforms both the classical $k$-means++ seeding and the single-swap local search from \cite{silvio-original}.

\vspace{-0.2cm}
\paragraph{Algorithms.}
The multi-swap local search algorithm that we analysed above performs very well in terms of solution quality. 
This empirically verifies the improved approximation factor of our algorithm, compared to the single-swap local search of \cite{silvio-original}.

Motivated by practical considerations, we heuristically adapt our algorithm to make it very competitive with SSLS in terms of running time and still remain very close, in terms of solution quality, to the theoretically superior algorithm that we analyzed.
The adaptation of our algorithm replaces the phase where it selects the $p$ centers to swap-out by performing an exhaustive search over $\binom{k + p}{p}$ subsets of centers. 
Instead, we use an efficient heuristic procedure for selecting the $p$ centers to swap-out, by greedily selecting one by one the centers to swap-out. 
Specifically, we select the first center to be the cheapest one to remove (namely, the one that increases the cost by the least amount once the points in its cluster are reassigned to the remaining centers). Then, we update all costs and select the next center iteratively. After $p$ repetitions we are done. 
We perform an experimental evaluation of the ``greedy'' variant of our algorithm compared to the theoretically-sound algorithm from \Cref{sec:multi-swap-kmpp} and show that employing the greedy heuristic does not measurably impact performance. 

The four algorithms that we evaluate are the following: 1) \textbf{\algokmeans:} The $k$-means++ from \cite{kmpp-original}, 2) \textbf{\algossls:} The Single-swap local search method from \cite{silvio-original}, 3) \textbf{\algomsls:} The multi-swap local search from Section~\ref{sec:multi-swap-kmpp}, and 4) \textbf{\algogreedymsls:} The greedy variant of multi-swap local search as described above. 

We use \algogreedymslsarg{x} and \algomslsarg{x}, to denote \algogreedymsls{} and \algomsls{} with $p=x$, respectively.
Notice that \algogreedymslsarg{1} is exactly \algossls.
Our experimental evaluation explores the effect of $p$-swap LS, for $p > 1$, in terms of solution cost and running time.

\vspace{-0.215cm}
\paragraph{Datasets.}
We consider the three datasets used in \cite{silvio-original} to evaluate the performance of SSLS: 1) KDD-PHY – $100,000$ points with $78$ features representing a quantum physic task \cite{kdd-datasets}, 2) RNA - $488,565$ points with $8$ features representing RNA input sequence pairs \cite{rna-dataset}, and 3) KDD-BIO – $145,751$ points with $74$ features measuring the match between a protein and a native sequence \cite{kdd-datasets}. We discuss the results for two or our datasets, namely KDD-BIO and RNA. We deffer the results on KDD-PHY to the appendix and note that the results are very similar to the results on RNA. 

We performed a preprocessing step to clean-up the datasets. We observed that the standard deviation of some features was disproportionately high. This causes all costs being concentrated in few dimensions making the problem, in some sense, lower-dimensional.  
Thus, we apply min-max scaling to all datasets and observed that this causes all our features' standard deviations to be comparable. 

\vspace{-0.215cm}
\paragraph{Experimental setting.} 
All our code is written in Python. The code will be made available upon publication of this work. We did not make use of parallelization techniques. To run our experiments, we used a personal computer with $8$ cores, a $1.8$ Ghz processor, and $15.9$ GiB of main memory 
We run all experiments 5 times and report the mean and standard deviation in our plots. All our plots report the progression of the cost either w.r.t local search steps, or Lloyd's iterations. 
We run experiments on all our datasets for $k=10, 25, 50$. The main body of the paper reports the results for $k=25$, while the rest can be found in the appendix. We note that the conclusions of the experiments for $k = 10, 50$ are similar to those of $k=25$. 

\vspace{-0.215cm}
\paragraph{Removing centers greedily.} 
We first we compare \algogreedymsls{} with \algomsls{}. To perform our experiment, we initialize $k=25$ centers using $k$-means++ and then run $50$ iterations of local search for both algorithms, for $p = 3$ swaps.
Due to the higher running of the \algomsls{} we perform this experiments on 1\% uniform sample of each of our datasets. 
We find out that the performance of the two algorithms is comparable on all our instances, while they both perform roughly 15\%-27\% at convergence.
\Cref{fig:vanilla-vs-greedy} shows the aggregate results, over 5 repetitions of our experiment.

\begin{figure}[t] 
    \centering
         \begin{subfigure}
         \centering
         \includegraphics[trim={0.6cm 0.8cm 0.2cm 1cm}, width=0.45\textwidth]{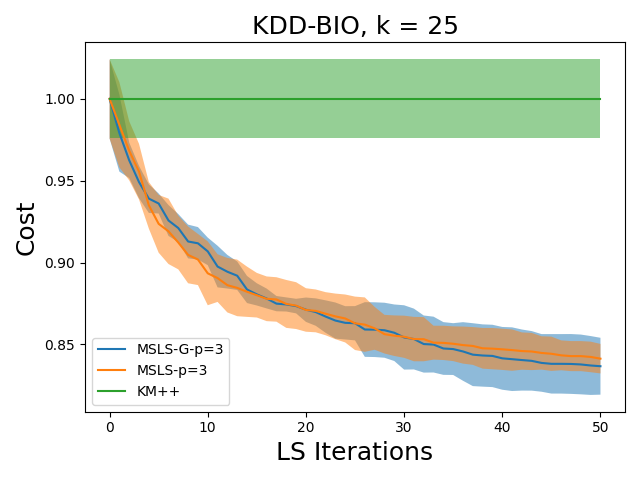}
    \end{subfigure}
    \begin{subfigure}
         \centering
        \includegraphics[trim={0.6cm 0.8cm 0.2cm 1cm}, width=0.45\textwidth]{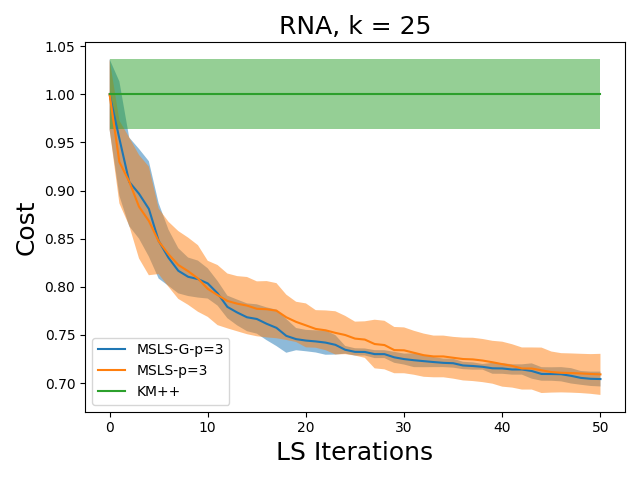}
    \end{subfigure}
       \caption{Comparison between \algomsls{} and \algogreedymsls{}, for $p = 3$, for $k=25$, on the datasets KDD-BIO and RNA. The $y$ axis shows the solution cost divided by the means solution cost of \algokmeans{}.}
       \iferrata
       \else
       \vspace{-0.5cm}
       \fi
    \label{fig:vanilla-vs-greedy}
\end{figure}

It may happen that \algomsls{}, which considers all possible swaps of size $p$ at each LS iteration, performs worse than \algogreedymsls{} as a sub-optimal swap at intermediate iterations may still lead to a better local optimum by coincidence.
Given that \algogreedymsls{} performs very comparably to \algomsls{}, while it is much faster in practice, we use \algogreedymsls{} for the rest of our experiments where we compare to baselines. This allows us to consider higher values of $p$, without compromising much the running time.

\vspace{-0.2cm}
\paragraph{Results: Evaluating the quality and performance of the algorithms.} 
In our first experiment we run \algokmeans{} followed by $50$ iterations of \algogreedymsls{} with $p = 1, 4, 7, 10$ and plot the relative cost w.r.t. \algokmeans{} at each iteration, for $k=25$. The first row of \Cref{fig:main-experiment} plots the results. Our experiment shows that, after $50$ iterations \algogreedymsls{} for $p = 4, 7, 10$ achieves improvements of roughly $10\%$ compared to \algogreedymslsarg{1} and of the order of $20\%-30\%$ compared to \algokmeans. We also report the time per iteration that each algorithm takes. For comparison, we report the running time of a single iteration of Lloyd's next to the dataset's name. It is important to notice that, although \algogreedymslsarg{1} is faster, running more iterations \algogreedymslsarg{1} is not sufficient to compete with \algogreedymsls{} when $p>1$.

\iferrata
\begin{figure} 
    \centering
     \begin{subfigure}
         \centering
         \includegraphics[trim={0.6cm 0cm 0.2cm 1cm},width=0.45\textwidth]{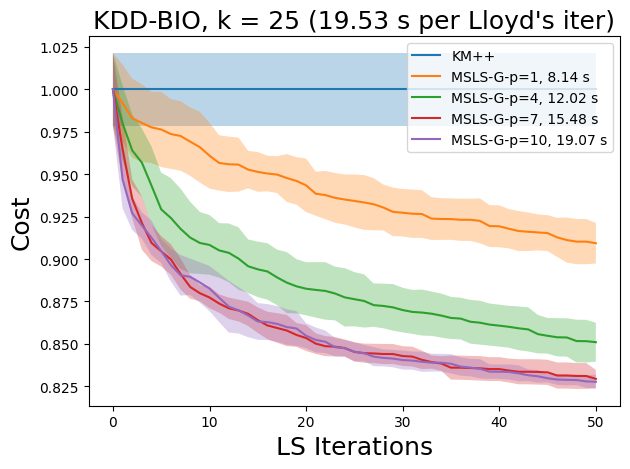}
    \end{subfigure}
         \begin{subfigure}
         \centering
         \includegraphics[trim={0.6cm 0cm 0.2cm 1cm},width=0.45\textwidth]{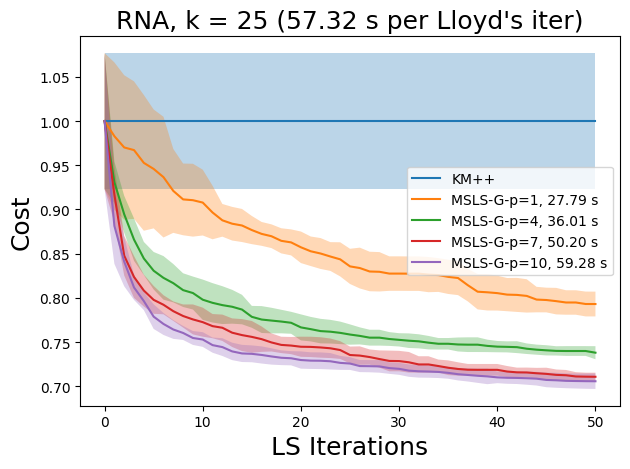}
    \end{subfigure}
    \begin{subfigure}
         \centering
        \includegraphics[trim={0.6cm 0.8cm 0.2cm 0.8cm},width=0.45\textwidth]{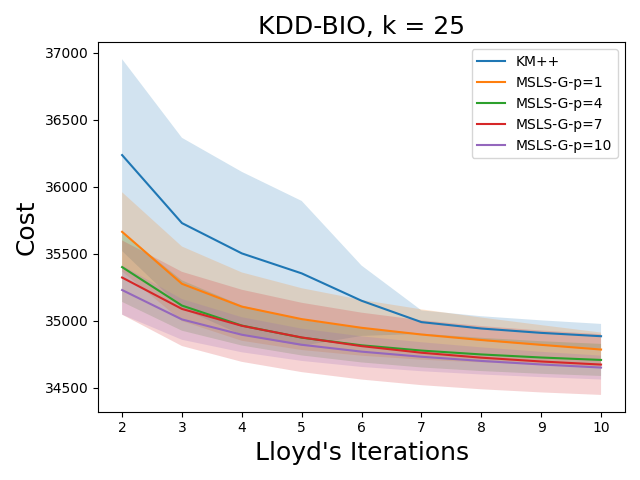}
    \end{subfigure}
         \begin{subfigure}
         \centering
         \includegraphics[trim={0.6cm 0.8cm 0.2cm 0.8cm},width=0.45\textwidth]{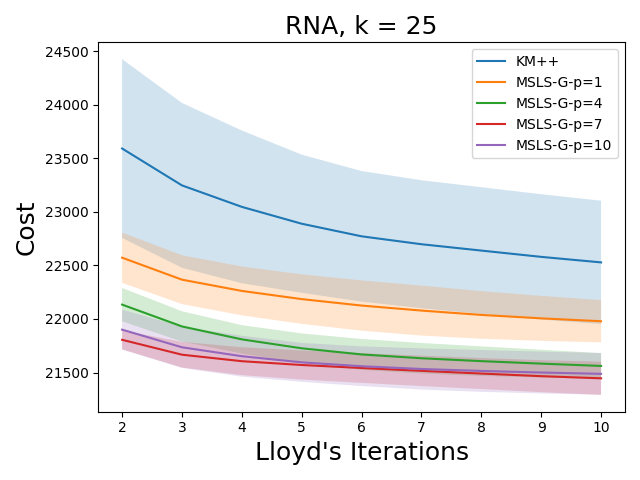}
    \end{subfigure}
        \caption{The first row compares the cost of \algogreedymsls{}, for $p\in\{1,4, 7, 10\}$, divided by the mean cost of \algokmeans{} at each LS  step, for $k=25$. The legend reports also the running time of \algogreedymsls{} per LS step (in seconds).
    The second row compares the cost after each of the $10$ iterations of Lloyd with seeding from \algogreedymsls{}, for $p\in\{1,4, 7, 10\}$ and $15$ local search steps and \algokmeans{}, for $k=25$.
    }
    \vspace{-0.5cm}
    \label{fig:main-experiment}
\end{figure}
\else
\begin{figure}[] 
    \centering
     \begin{subfigure}
         \centering
         \includegraphics[trim={0.6cm 0cm 0.2cm 1cm},width=0.45\textwidth]{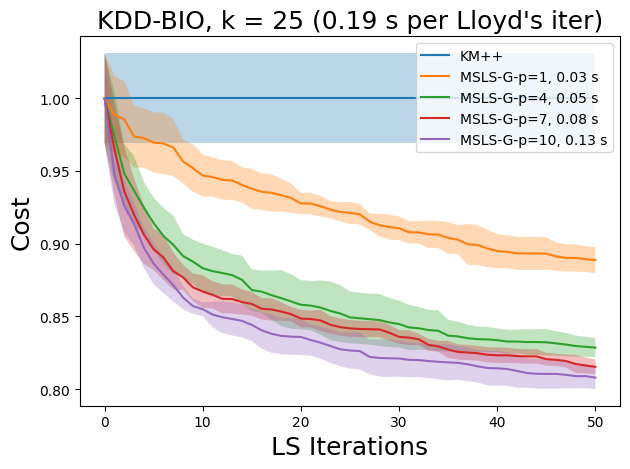}
    \end{subfigure}
         \begin{subfigure}
         \centering
         \includegraphics[trim={0.6cm 0cm 0.2cm 1cm},width=0.45\textwidth]{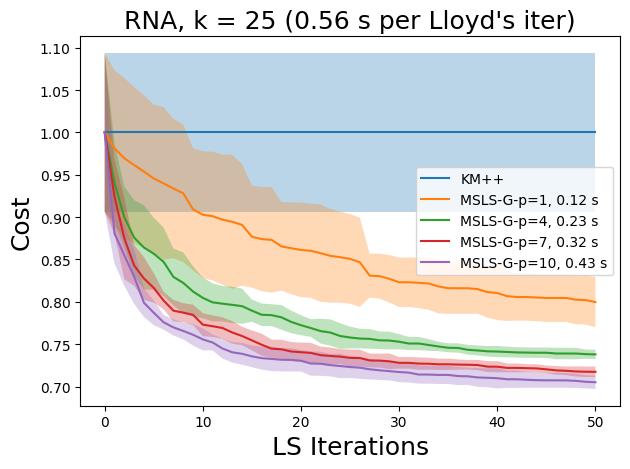}
    \end{subfigure}
    \begin{subfigure}
         \centering
        \includegraphics[trim={0.6cm 0.8cm 0.2cm 0.8cm},width=0.45\textwidth]{./images/lloyd-KDD-BIO.png}
    \end{subfigure}
         \begin{subfigure}
         \centering
         \includegraphics[trim={0.6cm 0.8cm 0.2cm 0.8cm},width=0.45\textwidth]{./images/lloyd-RNA.png}
    \end{subfigure}
        \caption{The first row compares the cost of \algogreedymsls{}, for $p\in\{1,4, 7, 10\}$, divided by the mean cost of \algokmeans{} at each LS  step, for $k=25$. The legend reports also the running time of \algogreedymsls{} per LS step (in seconds).  
    The second row compares the cost after each of the $10$ iterations of Lloyd with seeding from \algogreedymsls{}, for $p\in\{1,4, 7, 10\}$ and $15$ local search steps and \algokmeans{}, for $k=25$.
    }
    \label{fig:main-experiment}
\end{figure}
\fi

\vspace{-0.2cm}
\paragraph{Results: Evaluating the quality after postprocessing using Lloyd.}
In our second experiment, we use \algokmeans{} and \algogreedymsls{} as a seeding algorithm for Lloyd's and measure how much of the performance improvement measured in the first experiment is retained after running Lloyd's.
First, we initialize our centers using \algokmeans{} and the run $15$ iterations of \algogreedymsls{} for $p=1, 4, 7$.  We measure the cost achieved by running $10$ iterations of Lloyd's starting from the solutions found by \algogreedymsls{} as well as \algokmeans{}. In \Cref{fig:main-experiment} (second row) we plot the results.  
Notice that, according to the running times from the first experiment, $15$ iterations iterations of \algogreedymsls{} take less than $10$ iterations of Lloyd's for $p=4, 7$ (and also for $p=10$, except on RNA). We observe that \algogreedymsls{} for $p>1$ performs at least as good as \algossls{} from \cite{silvio-original} and in some cases maintains non-trivial improvements.

\paragraph{Results: Evaluating the quality and performance of the algorithms against a fixed deadline.} 
In this experiment we run \algokmeans{} followed by \algogreedymsls{} with $p = 1, 4, 7, 10$, for a set of fixed amounts of time.
This setting allows the versions of \algogreedymsls{} with smaller swap size to perform more iterations compared to the versions of the algorithm with a larger swap size, as smaller swap size leads to lower running time per iteration. 
Let $\tau$ be the average time that Lloyd's algorithm requires to complete a simple iteration on a specific instance.
We plot the cost of the solution produced by each algorithm after running $\lambda \times \tau$ for each $\lambda \in \{1, \cdots, 20\}$ in \Cref{fig:dealine-experiment}. 
Our experiment shows that \algogreedymsls{} for $p = 4, 7, 10$ achieves improvements of more than $5\%$ compared to \algogreedymslsarg{1} even when compared against a fixed running time, and of the order of $20\%-30\%$ compared to \algokmeans.

\begin{figure}[t] 
    \centering
         \begin{subfigure}
         \centering
         \includegraphics[trim={0.6cm 0.8cm 0.2cm 1cm}, width=0.49\textwidth]{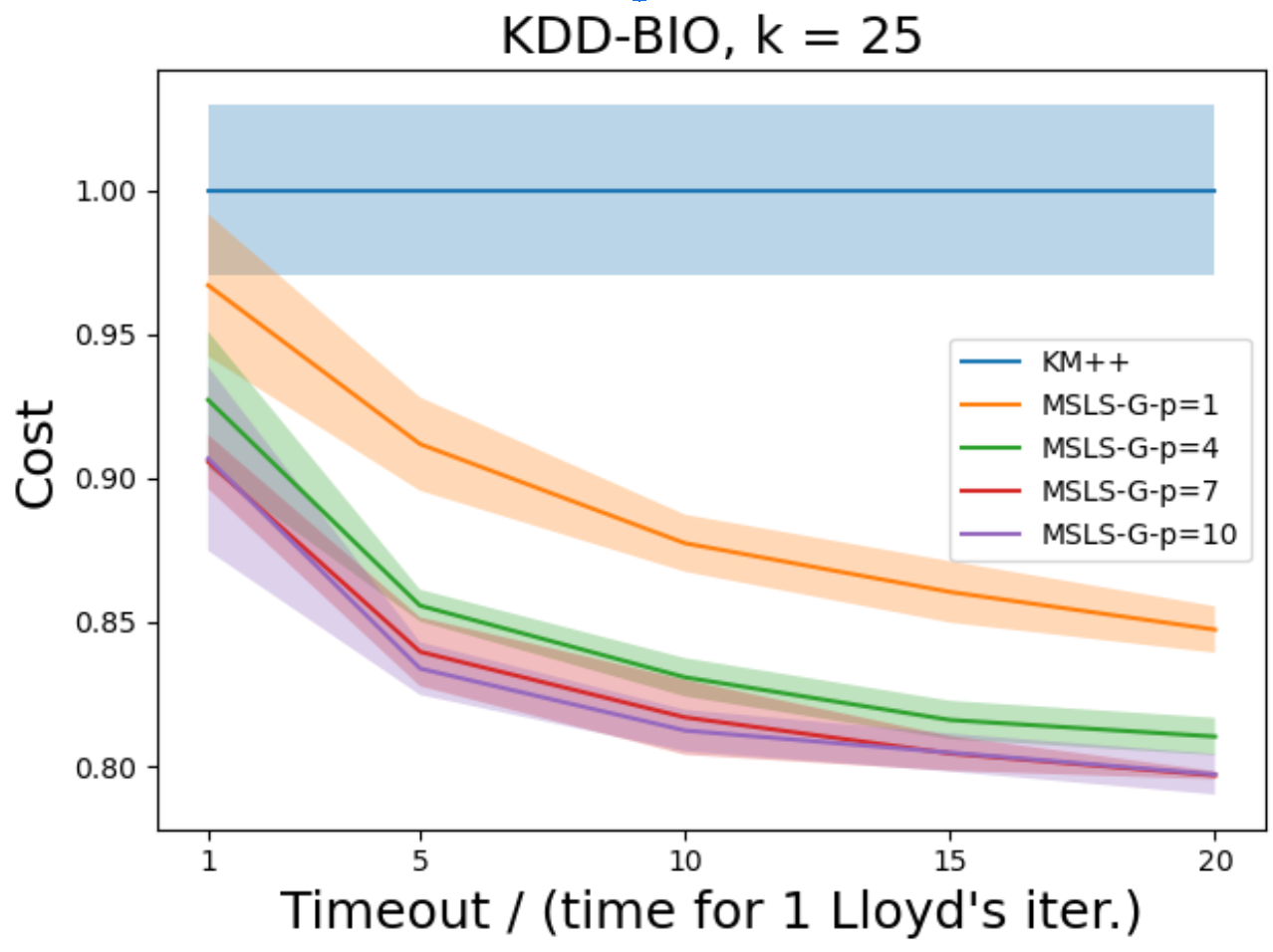}
    \end{subfigure}
    \begin{subfigure}
         \centering
        \includegraphics[trim={0.6cm 0.8cm 0.2cm 1cm}, width=0.49\textwidth]{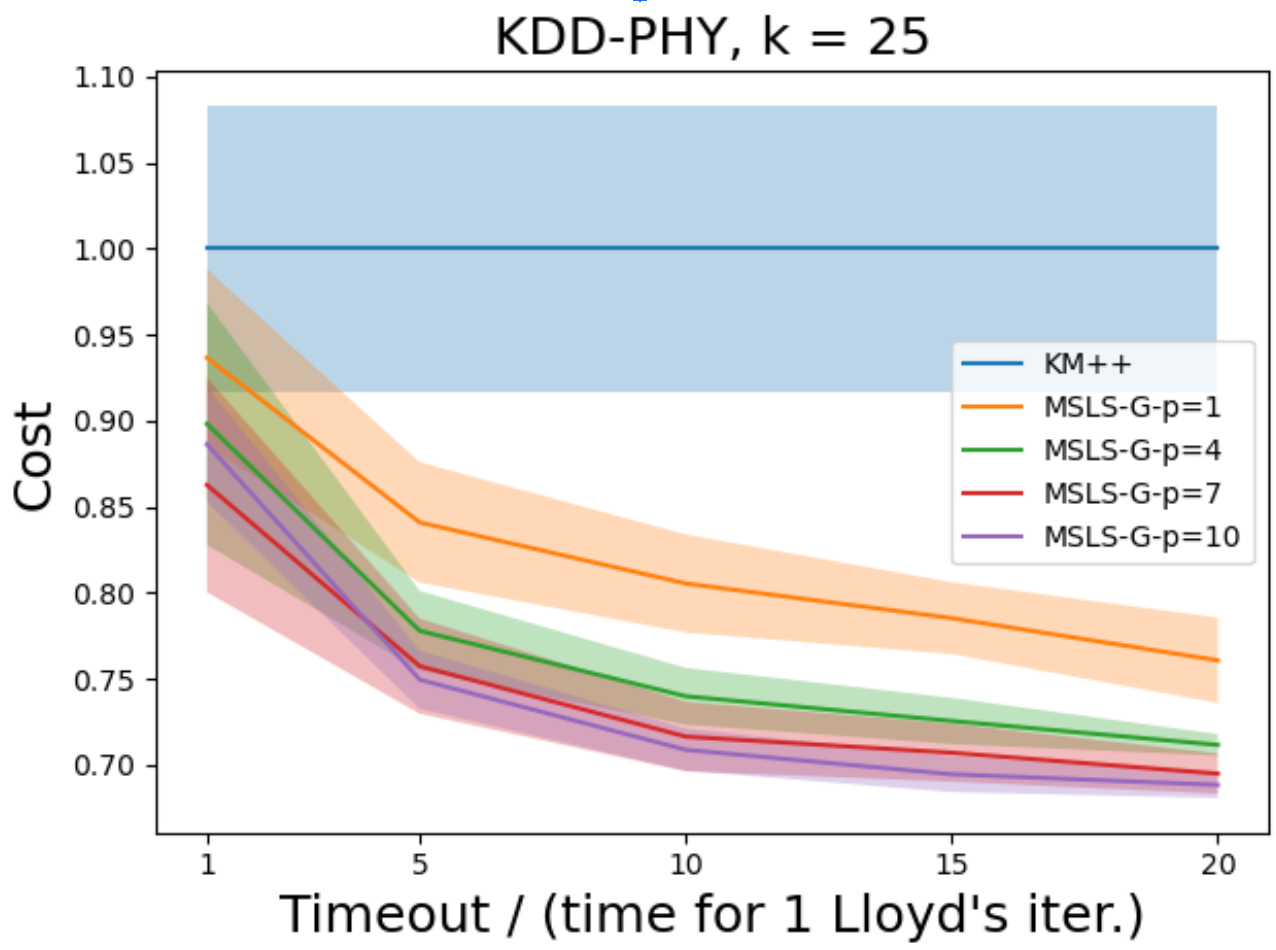}
    \end{subfigure}
       \caption{Comparison of the cost produced by \algogreedymsls{}, for $p\in\{1,4, 7, 10\}$ and $k=25$ on the datasets KDD-BIO and KDD-PHU, divided by the mean cost of \algokmeans{} after running for fixed amount of time in terms of multiplicative factors to the average time for an iteration of Lloyd's algorithm (i.e., for deadlines that are $1\times, \dots, 20\times$ the average time of an iteration of Lloyd).}
       \iferrata
       \else
       \vspace{-0.5cm}
       \fi
    \label{fig:dealine-experiment}
\end{figure}

\vspace{-0.25cm}
\section*{Conclusion and Future Directions}
\vspace{-0.2cm}
We present a new algorithm for the $k$-means problem and we show that it outperforms theoretically and experimentally state-of-the-art practical algorithms
with provable guarantees in terms of solution quality. A very interesting open question is to improve our local search procedure by avoiding the exhaustive search over all possible size-$p$ subsets of centers to swap out, 
concretely an algorithm with running time $\tilde{O}(2^{poly (1/ \varepsilon)} ndk)$.

\paragraph{Acknowledgements.}
This work was partially done when Lorenzo Beretta was a Research Student at Google Research. 
Moreover, Lorenzo Beretta receives funding from the European Union's Horizon 2020 research and innovation program under the Marie Skłodowska-Curie grant agreement No.~801199.
\includegraphics[height=0.03\textwidth, width=0.055\textwidth]{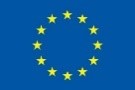}

\bibliographystyle{plainnat}
\bibliography{literature}

\newpage

{\centering \Large \textbf{Supplementary Material}}

\section*{Proofs from \Cref{sec:multi-swap-kmpp}}
\CombinedReassignmentCost*
\begin{proof}
\begin{align*}
\sum_{p \in P} \cost{p}{\calC[\optset[p]]} &= \\
\sum_{o_i \in \optset} \sum_{p \in \optcl_i} \cost{p}{\calC[o_i]} &= \\
\sum_{o_i \in \optset} \abs{\optcl_i} \cdot \cost{o_i}{\calC[o_i]} + \cost{\optcl_i}{o_i} &= \\
\opt + \sum_{p \in P} \cost{\optset[p]}{\calC[\optset[p]]} &\leq \\
\opt + \sum_{p \in P} \cost{\optset[p]}{\calC[p]} &\leq \\
\opt + \sum_{p \in P} \ld( \norm{\optset[p]- p} + \norm{p- \calC[p]}\rd)^2 &= \\
2\opt + \alg + 2 \sum_{p\in P} \norm{\optset[p], p} \cdot \norm{p, \calC[p]} &\leq 2\opt + \alg + 2\sqrt{\alg} \sqrt{\opt}.
\end{align*}
The second equality is due to \Cref{lem:moved-cost} and the last inequality is due to Cauchy-Schwarz.
\end{proof}

\section*{Proofs from \Cref{sec:faster-9+eps-algo}}

In this section, we prove the following.
\MainTheorem*

We start with a key lemma showing that a sample of size $O(1 /\varepsilon)$ is enough to approximate $1$-mean.

\begin{lemma}[Form \cite{easy-coreset}] \label{lem:weak-coreset-for-1-mean}
Given an instance $P \subseteq \bbR^d$, sample $m = 1 / (\varepsilon \delta)$ points uniformly at random from $P$ and denote the set of samples with $S$. Then $\cost{P}{\mu(S)} \leq (1+\varepsilon) \cost{P}{\mu(P)}$ with probability at least $1-\delta$.
\end{lemma}
\begin{proof}
We want to prove that with probability $1 - \delta$ we have $||\mu(S) - \mu(P)||^2 \leq \varepsilon \cost{P}{\mu(P)}/ |P|$. Then, applying \Cref{lem:moved-cost} gives the desired result.
First, we notice that $\mu(P)$ is an unbiased estimator of $\mu(P)$, namely $E[\mu(S)] = \mu(P)$.
Then, we have 
\begin{align*}
	E\ld[||\mu(S) - \mu(P)||^2\rd] = 
	\frac{1}{m} \sum_{i=1}^{|S|} E\ld[||s_i - \mu(P)||^2 \rd] = 
	\frac{\cost{P}{\mu(P)}}{m \cdot |P|}
\end{align*}
where $s_i$ are uniform independent samples from $P$. Applying Markov's inequality concludes the proof.
\end{proof}

The algorithm that verifies \Cref{thm:main-result} is very similar to the MSLS algorithm from \Cref{sec:multi-swap-kmpp} and we use the same notation to describe it. The intuition is that in MSLS we sample $\calQ = \{q_1 \dots q_p\}$ hoping that $q_i \in \core{o_i}$ for each $i$; here we refine $q_i$ to a better approximation $\ho_i$ of $o_i$ and swap the points $(\ho_i)_i$ rather than $(q_i)_i$. Our points $\ho_i$ are generated taking the average of some sampled point, thus we possibly have $\ho_i \not\in P$ while, on the other hand, $q_i \in P$.   

\paragraph{A $(9+\varepsilon)$-approximation MSLS algortihm.} 
First, we initialize our set of centers using $k$-means++. Then, we run $ndk^{O(\varepsilon^{-2})} \cdot 2^{poly(\varepsilon^{-1})}$ local search steps, where a local search step works as follows. Set $p = \Theta(\varepsilon^{-1})$. We $D^2$-sample a set $\calQ = \set{q_1 \dots q_p}$ of points from $P$ (without updating costs). Then, we iterate over all possible sets $Out = \set{c_1 \dots c_p}$ of $p$ distinct elements in $\calC \cup \calQ$. We define the set of \emph{temporary} centers $\calT = (\calC \cup \calQ ) \setminus Out$ and run a subroutine $\approxcenters(\calT)$ which returns a list of $poly(\varepsilon^{-1}) \cdot \log^{O(\varepsilon^{-1})}(\Delta)$ size-$s$ sets $\widehat{In} = \{\ho_1 \dots \ho_s\}$ (where $s = |\calQ \setminus Out|$). We select the set $\whin$ in this list such that the swap $(\whin, Out \setminus \calQ)$ yields the maximum cost reduction. Then we select the set $Out$ that maximizes the cost reduction obtained in this way. If $(\widehat{In}, Out \setminus \calQ)$ actually reduces the cost then we perform that swap. 

\paragraph{A subroutine to approximate optimal centers.}
Here we describe the subroutine $\approxcenters(\calT)$. Let $\calQ \setminus Out = \{q_1 \dots q_s\}$. Recall that $s \leq p = O(\varepsilon^{-1})$. This subroutine outputs a list of $2^{poly(\varepsilon^{-1})} \cdot \log^{O(\varepsilon^{-1})}(\Delta)$ size-$s$ sets $\widehat{In} = \{\ho_1 \dots \ho_s\}$.
Here we describe how to find a list of $2^{poly(\varepsilon^{-1})} \cdot \log(\Delta)$ values for $\ho_1$. The same will apply for $\ho_2 \dots \ho_s$ and taking the Cartesian product yields a list of $2^{poly(\varepsilon^{-1})} \cdot \log^{O(\varepsilon^{-1})}(\Delta)$ size-$s$ sets.
Assume wlog that the pairwise distances between points in $P$ lie in $[1, \Delta]$. We iterate over all possible values of $\rho_1 \in \{1, (1+\varepsilon) \dots (1+\varepsilon)^{\lceil \log_{1+\varepsilon}\Delta \rceil}\}$. We partition $P$ in three sets: the set of \textit{far points} $F = \{ x \in P \,|\, cost(x, q_1) > \rho_1^2 / \varepsilon^3\}$, the set of \textit{close points} $C = \{ x \in P \setminus F \,|\, cost(x, \calT) \leq \varepsilon^3 \rho_1^2\}$ and the set of \textit{nice points} $N = P \setminus (C \cup F)$.
Then, we sample uniformly from $N$ a set $S$ of size $\Theta(\varepsilon^{-1})$. For each $(s+1)$-tuple of coefficients $\alpha_0, \alpha_1 \dots \alpha_s \in \set{1, (1 - \varepsilon), (1 - \varepsilon)^2,  \dots (1-\varepsilon)^{\lceil\log_{1-\varepsilon}(\varepsilon^{7})\rceil}} \cup \{0\}$ we output the candidate solution given by the convex combination

\begin{equation} \label{eq:def-ohat}
	\ho_1 = \ho_1(\alpha_0 \dots \alpha_s) = \frac{\alpha_0 \mu(S) + \sum_{i=1}^s \alpha_i q_i}{\sum_{i =0}^s \alpha_i}
\end{equation}
so, for each value of $\rho_1$, we output $2^{poly(\varepsilon^{-1})}$ values for $\ho_1$. Hence, $2^{poly(\varepsilon^{-1})} \cdot \log(\Delta)$ values in total.

\subsection*{Analysis} 

The key insight in the analysis of the MSLS algorithm form \Cref{sec:multi-swap-kmpp} was that every $q_i$ was a proxy for $o_i$ because $q_i \in \core{o_i}$, and thus $q_i$ provided a good center for $\optcl_i$. In the analysis of this improved version of MSLS we replace $q_i$ with $\ho_i$ which makes a better center for $\optcl_i$.   
Formally, fixed $Out$, we say that a point $\ho_i$ is a \emph{perfect approximation} of $o_i$ when $\cost{\optcl_i}{(\calC \cup \{\ho_i\}) \setminus  Out} \leq (1 + \varepsilon) \opt_i + \varepsilon \opt / k$. 
We define $\tilo$ and $\tilc$ as in \Cref{sec:multi-swap-kmpp}, except that we replace $\delta$ with $\varepsilon$ (which here is not assumed to be a constant). Likewise, we build the set $\calS$ of ideal multi-swaps as in \Cref{sec:multi-swap-kmpp}.
Recall that we say that a multi-swap $(In, Out)$ is \emph{strongly improving} if $\cost{P}{(\calC \cup In) \setminus Out} \leq (1 - \varepsilon/ k ) \cdot \cost{P}{\calC}$. 
Let $In = \{o_1 \dots o_s\} \subseteq \tilo$ and $Out = \{c_1 \dots c_s \} \subseteq \tilc$, we overload the definition from \Cref{sec:multi-swap-kmpp} and say that the ideal multi-swap $(In, Out)$ is \emph{good} if for every $\whin = \set{\ho_1 \dots \ho_s}$ such that each $\ho_i$ is a perfect approximation of $o_i$ for each $i = 1 \dots s$ the swap $(\whin, Out)$ is strongly improving. We call an ideal swap \emph{bad} otherwise. 
As in \Cref{sec:multi-swap-kmpp}, we define the \emph{core} of an optimal center; once again we replace $\delta$ with $\epsilon$, which is no longer constant.
The two following lemmas are our stepping stones towards \Cref{thm:main-result}.

\begin{lemma} \label{lem:qs-in-cores-and-good-swap}
If $\alg / \opt > 9 + O(\varepsilon)$ then, with probability $k^{-O(\varepsilon^{-1})} \cdot 2^{-poly(\varepsilon^{-1})}$, there exists $Out \subseteq \calC \cup \calQ$ such that:
\begin{enumerate}[(i)]
    \item If $\calQ \setminus Out = \{q_1 \dots q_s\}$ then $q_1 \in \core{o_1} \dots q_s \in \core{o_s}$ for some $o_1 \dots o_s \in \optset$
    \item If we define $In = \{o_1 \dots o_s\}$ then $(In, Out \setminus \calQ)$ is a good ideal swap.
\end{enumerate}  
\end{lemma}

\begin{lemma} \label{lem:perfect-approximations}
If $(i)$ from \Cref{lem:qs-in-cores-and-good-swap} holds, then with probability $k^{-O(\varepsilon^{-2})} \cdot 2^{-poly(\varepsilon^{-1})}$, the list returned by $\approxcenters$ contains $\whin = \{\ho_1 \dots \ho_s\}$ such that $\ho_i$ is a perfect approximation of $o_i$ for each $i = 1 \dots s$.
\end{lemma}

\begin{proof}[Proof of \Cref{thm:main-result}.]
Here we prove that our improved MSLS algorithm achieves a $(9 + O(\varepsilon))$-approximation, which is equivalent to \Cref{thm:main-result} up to rescaling $\varepsilon$.
Combining \Cref{lem:qs-in-cores-and-good-swap} and \Cref{lem:perfect-approximations} we obtain that, as long as $\alg/\opt > 9 + O(\varepsilon)$, with probability at least $k^{-O(\varepsilon^{-2})} \cdot 2^{-poly(\varepsilon^{-1})}$, the list returned by $\approxcenters$ contains $\whin = \{\ho_1 \dots \ho_s\}$ such that $(\whin, Out \setminus \calQ)$ is strongly improving. If this happens, we call such a local step \emph{successful}. Now the proof goes exactly as the proof of \Cref{thm:multi-swap-analysis}. Indeed, We show that 
$k^{O(\varepsilon^{-2})} \cdot 2^{poly(\varepsilon^{-1})}$ local steps suffice to obtain $\Omega(k \log\log k / \varepsilon)$ successful local steps, and thus to obtain the desired approximation ratio, with constant probability.

To prove the running time bound it is sufficient to notice that a local search step can be performed in time $nd\log^{O(\varepsilon^{-1})}(\Delta) \cdot 2^{poly(\varepsilon^{-1})}$.   
\end{proof}

In the rest of this section, we prove \Cref{lem:qs-in-cores-and-good-swap} and \Cref{lem:perfect-approximations}.
\begin{observation} \label{obs:sample-from-core++}
If we assume $\delta = \varepsilon$ non-constant in \Cref{lem:sample-from-core}, then performing the computations explicitly we obtain $\Prp{q \in \core{\optcl_i}} \geq poly(\varepsilon)$.
\end{observation}

In order to prove \Cref{lem:qs-in-cores-and-good-swap}, we first prove the two lemmas. \Cref{lem:bad-swap-inequality++} is the analogous of \Cref{lem:bad-swap-inequality} and \Cref{lem:all-qs-in-cores++} is the analogous of \Cref{lem:all-qs-in-cores}. Overloading once again the definition from \Cref{sec:multi-swap-kmpp}, we define $G$ as the union of cores of good optimal centers in $\tilo$, where an optimal center is defined to be good if at least one of the ideal multi-swaps in $\calS$ it belongs to is good (exactly as in \Cref{sec:multi-swap-kmpp}).

\begin{lemma} \label{lem:bad-swap-inequality++}
If an ideal swap $(In, Out)$ is bad, then we have
\begin{equation} \label{eq:bad-swap-ineq}
\cost{\optcl_{In}}{\calC} \leq (1+\varepsilon) \cost{\optcl_{In}}{\optset} + \reassign{In}{Out} + \varepsilon \alg / k.
\end{equation}
\end{lemma}
\begin{proof}
Let $In = \{o_1 \dots o_s\}$, $\whin = \set{\ho_1 \dots \ho_s}$ such that $\ho_i$ is a perfect approximation of $o_i$ for each $i = 1 \dots s$. Recall that $\optcl_{In} := \bigcup_{i=1}^s \optcl_i$, then 
\begin{equation}
\cost{\optcl_{In}}{(\calC \cup \widehat{In}) \setminus Out} \leq \sum_{i=1}^s \cost{\optcl_i}{(\calC \cup \{\ho_i\}) \setminus Out} \leq (1+\varepsilon) \cost{\optcl_{In}}{\optset}.
\end{equation}

Moreover, $\reassign{In}{Out} = \cost{P \setminus \optcl_{In}}{\calC \setminus Out} - \cost{P \setminus \optcl_{In}}{\calC}$ because points in $P \setminus C_{Out}$ are not affected by the swap. 
Therefore, $\cost{P}{(\calC \cup \whin) \setminus Out} \leq (1 + \varepsilon) \cost{\optcl_{In}}{\optcl} + \reassign{In}{Out} + \cost{P \setminus \optcl_{In}}{\calC}$. Suppose by contradiction that \Cref{eq:bad-swap-ineq} does not hold, then
\begin{align*}
\cost{P}{\calC} - \cost{P}{(\calC \cup \whin) \setminus Out} &= \\ \cost{P \setminus \optcl_{In}}{\calC} + \cost{\optcl_{In}}{\calC} - \cost{P}{(\calC \cup \whin) \setminus Out} &\geq \epsilon \alg / k.
\end{align*}
Hence, $(\whin, Out)$ is strongly improving and this holds for any choice of $\whin$, contradiction.
\end{proof}

\begin{lemma} \label{lem:all-qs-in-cores++}
If $\alg  /\opt > 9 + O(\varepsilon)$ then $\cost{G}{\calC} \geq \cost{P}{\calC} \cdot poly(\varepsilon)$. Thus, if we $D^2$-sample $q$ we have $P[q \in G] \geq poly(\varepsilon)$.
\end{lemma}
\begin{proof}
First, we observe that the combined current cost of all optimal clusters in $\optset \setminus \tilo$ is at most $k \cdot \varepsilon \alg / k = \varepsilon \alg$. 
Now, we prove that the combined current cost of all $\optcl_i$ such that $o_i$ is bad is $\leq (1-2\varepsilon)\alg$. Suppose, by contradiction, that it is not the case, then we have:
\begin{align*} \label{eq:combined-opt-reassignment}
(1-2\varepsilon)\alg < \sum_{\text{ Bad } o_i \in \tilo} \cost{\optcl_i}{\calC} \leq 
\sum_{\text{ Bad } (In, Out) \in \calS} w(In, Out) \cdot \cost{\optcl_{In}}{\calC} &\leq \\
\sum_{\text{ Bad } (In, Out)} w(In, Out) \cdot \ld( (1 + \varepsilon) \cost{\optcl_{In}}{\optset} + \reassign{In}{Out} + \varepsilon \alg / k \rd) &\leq \\
(1+ \varepsilon) \opt + (2 + 2/p)\opt + (2 + 2/p)\sqrt{\alg}\sqrt{\opt} + \varepsilon \alg.
\end{align*}

The second and last inequalities make use of \Cref{obs:ideal-swaps-weights}. The third inequality uses \Cref{lem:bad-swap-inequality++}.

Setting $\eta^2 = \alg / \opt$ we obtain the inequality $\eta^2 -(2 + 2/p \pm O(\varepsilon)) \eta - (3 + 2/p \pm O(\varepsilon)) \leq 0$.
Hence, we obtain a contradiction in the previous argument as long as $\eta^2 -(2 + 2/p \pm O(\varepsilon)) \eta - (3 + 2/p \pm O(\varepsilon)) > 0$, which holds for $p = \Theta(\varepsilon^{-1})$ and $\eta^2 = 9 + O(\varepsilon)$. A contradiction there implies that at least an $\varepsilon$-fraction of the current cost is due to points in $\bigcup_{\text{Good } o_i \in \tilo} \optcl_i$. 
Thanks to \Cref{obs:sample-from-core++}, we have $P_{q \sim \cost{q}{\calC}}[q \in \core{\optcl_i} \cond q \in \optcl_i] \geq poly(\varepsilon)$. Therefore, we can conclude that the  current cost of $G = \bigcup_{\text{Good } o_i \in \tilo} \core{\optcl_i}$ is at least a $poly(\varepsilon)$-fraction of the total current cost.  
\end{proof}

\begin{proof}[Proof of \Cref{lem:qs-in-cores-and-good-swap}.]
Thanks to \Cref{lem:all-qs-in-cores++}, we have that $P[q_1 \in G] \geq poly(\varepsilon)$. 
Whenever $q_1 \in G$ we have that $q_1 \in \core{o_1}$ for some good $o_1$. Then, for some $s \leq p$ we can complete $o_1$ with $o_2 \dots o_s$ such that $In  = \set{o_1 \dots o_s}$ belongs to a good swap. Concretely, there exists $Out \subseteq \calC$ such that $(In, Out)$ is a good swap. Since $In \subset \tilo$ we have $\cost{\optcl_i}{\calC} > \varepsilon \opt / k$ for all $o_i \in In$, which combined with \Cref{obs:sample-from-core++} gives that, for each $i = 2 \dots s$, $P[q_i \in \core{o_i}] \geq poly(\varepsilon) / k$. Hence, we have $P[q_i \in \core{o_i} \text{ for } i = 1 \dots s] \geq 2^{-poly(\varepsilon^{-1})} k^{-O(\varepsilon^{-1})}$.
Notice, however, that $(\whin, Out)$ is a $s$-swap and we may have $s < p$. Nevertheless, whenever we sample $q_1 \dots q_s$ followed by any sequence $q_{s+1} \dots q_p$ it is enough to choose $Out' = Out \cup \{q_{s+1} \dots q_p\}$ to obtain that $(\{q_1 \dots q_p\}, Out')$ is an improving $p$-swap.
\end{proof}

In order to prove \Cref{lem:perfect-approximations} we first need a few technical lemmas.
\begin{lemma}[Lemma 2 from \cite{silvio-original}] \label{lem:lattanzi-2}
For each $x, y, z \in \bbR^d$ and $\varepsilon > 0$, $\cost{x}{y} \leq (1+\varepsilon) \cost{x}{z} + (1 + 1/\varepsilon) \cost{z}{y}$.  
\end{lemma}

\begin{lemma} \label{lem:as-if-they-were-concentrated}
Given $q \in \bbR^d$ and $Z \subseteq \bbR^d$ such that $\cost{Z}{q}\leq \varepsilon^2 \Gamma$ then, for each $o \in \bbR^d$
\begin{equation*}
(1 - O(\varepsilon)) \cost{Z}{o} - O(\varepsilon) \Gamma\leq |Z| \cost{q}{o} \leq  (1 + O(\varepsilon)) \cost{Z}{o} + O(\varepsilon) \Gamma   
\end{equation*}
\end{lemma}
\begin{proof}
To obtain the first inequality, we apply \Cref{lem:lattanzi-2}  to bound $\cost{z}{o} \leq (1 + \varepsilon) \cost{z}{o} + (1 + 1/\varepsilon) \cost{z}{q}$ for each $z \in Z$. To obtain the second inequality, we bound $\cost{q}{o} \leq (1 + \varepsilon) \cost{z}{o} + (1 + 1/\varepsilon) \cost{z}{q}$ for each $z \in Z$.
\end{proof}

\begin{lemma} \label{lem:alpha-are-off-its-ok}
Let $X = \{x_1 \dots x_\ell\}$ be a weighted set of points in $\bbR^d$ such that $x_i$ has weight $w_i$. Let $\mu$ be the weighted average of $X$. 
Let $\hat{\mu} = \hat{\mu}(\alpha_1 \dots \alpha_\ell)$ be the weighted average of $X$ where $x_i$ has weight $\alpha_i$.
If $w_i \leq \alpha_i \leq w_i / (1-\varepsilon)$ for each $i = 1 \dots \ell$, then if we interpret $\cost{X}{C}$ as $\sum_{x_i \in X} w_i \cdot \cost{x_i}{C}$ we have 
$\cost{X}{\hat\mu} \leq (1 + O(\varepsilon)) \cost{X}{\mu}$. 
\end{lemma}
\begin{proof}
We note that $\mu$ minimizes the expression $\cost{X}{\mu}$. Moreover, $\cost{X}{z} \leq \sum_{i=1}^\ell \alpha_i \cdot \cost{x_i}{z} \leq  \cost{X}{z} / (1-\varepsilon)$. Since $\hat\mu$ minimizes the expression $\sum_{i=1}^\ell \alpha_i \cdot \cost{x_i}{z}$ it must be $\cost{X}{\hat\mu} \leq \cost{X}{\mu} / (1-\varepsilon)$.  
\end{proof}
Adopting the same proof strategy, we obtain the following.
\begin{observation} \label{obs:as-if-concentrated-obs}
Thanks to \Cref{lem:as-if-they-were-concentrated}, we can assume that the points in $Z$ are concentrated in $q$ for the purpose of computing a $(1 + O(\varepsilon))$-approximation to the $1$-means problem on $Z$, whenever an additive error $\Gamma$ is tolerable. Indeed, moving all points in $Z$ to $q$ introduces a $1+O(\varepsilon)$ multiplicative error on $\cost{Z}{\cdot}$ and a $O(\varepsilon) \Gamma$ additive error.
\end{observation}

The next lemma shows that a point $z$ that is far from a center $o$ experiences a small variation of $\cost{z}{o}$ when the position of $o$ is slightly perturbed. 
\begin{lemma} \label{lem:gradient-argument}
Given $o, z \in \bbR^d$ such that $||o-z|| \geq r / \varepsilon$ we have that for every $o' \in B(o, r)$, $\cost{z}{o'} = (1 \pm O(\varepsilon))\cost{z}{o}$.
\end{lemma}
\begin{proof}
It is enough to prove it for all $o'$ that lie on the line $L$ passing through $o$ and $z$, any other point in $o'' \in B(o, r)$ admits a point $o' \in B(o, r) \cap L$ with $||o' - z|| = ||o'' - z||$. It is enough to compute the derivative of $\cost{z}{\cdot}$ with respect to the direction of $L$ and see that $\frac{\partial \cost{z}{\cdot}}{\partial L}|_{B(o, r)} = (1\pm O(\varepsilon)) r / \varepsilon$. Thus, $\cost{z}{o'} = \cost{z}{o} \pm (1\pm O(\varepsilon)) r^2/\varepsilon = (1 \pm O(\varepsilon)) \cost{z}{o}$.
\end{proof}

\begin{proof}[Proof of \Cref{lem:perfect-approximations}]
Here we prove that for each $o_1 \dots o_s$ there exist coefficients $\alpha^{(i)}_0 \dots \alpha^{(i)}_s \in \set{1, (1-\varepsilon) \dots (1-\varepsilon)^{\lceil \log_{1-\varepsilon}(\varepsilon^{7})} \rceil} \cup \{0\}$ such that the convex combination $\ho_i = \ho_i(\alpha^{(i)}_0 \dots \alpha^{(i)}_s)$ is a perfect approximation of $o_i$, with probability $k^{-O(\varepsilon^{-2})} \cdot 2^{-poly(\varepsilon^{-1})}$.  Wlog, we show this for $o_1$ only. Concretely, we want to show that, with probability $k^{-O(\varepsilon^{-1})} \cdot 2^{-poly(\varepsilon^{-1})}$, there exist coefficients $\alpha_0 \dots \alpha_s$ such that $\ho_1=\ho_1(\alpha_0 \dots \alpha_s)$ satisfies $\cost{\optcl_1}{(\calC \cup \{\ho_1\}) \setminus  Out} \leq (1 + O(\varepsilon))\opt_1 + O(\varepsilon)\opt / k $. Taking the joint probability of these events for each $i=1\dots s$ we obtain the success probability $k^{-O(\varepsilon^{-2})} \cdot 2^{-poly(\varepsilon^{-1})}$. Note that we are supposed to prove that $\cost{\optcl_1}{(\calC \cup \{\ho_1\}) \setminus  Out} \leq (1 + \varepsilon) \opt_1 + \varepsilon \opt /k$, however we prove a weaker version where $\varepsilon$ is replaced by $O(\varepsilon)$, which is in fact equivalent up to rescaling $\varepsilon$.

Similarly to $\calC[\cdot]$ and $\optset[\cdot]$ define $\calT[p]$ as the closest center to $p$ in $\calT$. Denote with $C_1, F_1$ and $N_1$ the intersections of $\optcl_1$ with $C, F$ and $N$ respectively. In what follows we define the values of $\alpha_0 \dots \alpha_s$ that define $\ho_1 = \ho_1(\alpha_0 \dots \alpha_s)$ and show an assignment of points in $\optcl_1$ to centers in $(\calC \cup \{\ho_1\}) \setminus  Out$ with cost $(1 + O(\varepsilon))\opt_1 + O(\varepsilon)\opt / k $. Recall that we assume that $q_i \in \core{o_i}$ for each $i = 1 \dots s$.

In what follows, we assign values to the coefficients $(\alpha_i)_i$. It is understood that if the final value we choose for $\alpha_i$ is $v$ then we rather set $\alpha_i$ to the smallest power of $(1-\varepsilon)$ which is larger than $v$, if $v > \varepsilon^{7}$. Else, set $\alpha_i$ to $0$. We will see in the end that this restrictions on the values of $\alpha_i$ do not impact our approximation.

In what follows, we will assign the points in $\optcl_1$ to $\calC \setminus Out$, if this can be done inexpensively. If it cannot, then we will assign points to $\ho_1$. In order to compute a good value for $\ho_1$ we need an estimate of the average of points assigned to $\ho_1$. For points in $N_1$, computing this average is doable (leveraging \Cref{lem:weak-coreset-for-1-mean}) while for points in $\optcl_1 \setminus N_1$ we show that either their contribution is negligible or we can collapse them so as to coincide with some $q_i \in \calQ$ without affecting our approximation.
The coefficients $(\alpha_i)_{i\geq 1}$ represent the fraction of points in $\optcl_i$ which is collapsed to $q_i$. $\alpha_0$ represents the fraction of points in $\optcl_i$ which average we estimate as $\mu(S)$. Thus, \Cref{eq:def-ohat} defines $\ho_i$ as the weighted average of points $q_i$, where the weights are the (approximate) fractions of points collapsed onto $q_i$, together with the the average $\mu(S)$ and its associated weight $\alpha_0$.

\paragraph{Points in $C_1$.} All points $p \in C_1$ such that $\calT[p] \not \in \calQ$ can be assigned to $\calT[p] \in \calC \setminus Out$ incurring a total cost of at most $\varepsilon^6 \opt_1$, by the definition of $C_1$. 
Given a point $p \in C_1$ with $\calT[p] \in \calQ$ we might have $\calT[p] \not \in \calC \setminus Out$ and thus we cannot assign $p$ to $\calT[p]$. 
Denote with $W$ the set of points $p$ with $\calT[p] \in \calQ$. Our goal is now to approximate $\mu(W)$.
In order to do that, we will move each $p \in W$ to coincide with $q_i = \calT[p]$. We can partition $W$ into $W_1 \dots W_s$ so that for each $z \in W_i$ $\calT[z] = q_i$. If $p \in Z_i$ then we have $||p-q_i||^2 \leq \varepsilon^3 \rho_1^2$. 
Hence, thanks to \Cref{obs:as-if-concentrated-obs}, we can consider points in $W_i$ as if they were concentrated in $q_i$ while losing at most an additive factor $O(\varepsilon) \opt_1$ and a multiplicative factor $(1+\varepsilon)$ on their cost.
For $i = 1 \dots s$, set $\alpha_i \leftarrow |W_i| / |\optcl_1|$.
In this way, $\sum_{i=1}^s \alpha_i \cdot q_i / \sum_{i=1}^s \alpha_i$ is an approximates solution to $1$-mean on $W$ up to a multiplicative factor $(1+\varepsilon)$ and an additive factor $O(\varepsilon) \opt_1$.

\paragraph{Points in $N_1$.}
Consider the two cases: $(i)$ $\cost{N_1}{\calT} > \varepsilon^2 \opt / k$; $(ii)$ $\cost{N_1}{\calT} \leq \varepsilon^2 \opt / k$.

Case $(i)$. We show that in this case $\mu(S)$ is a $(1+\varepsilon)$-approximation for $1$-mean on $N_1$, with probability $k^{-O(\varepsilon^{-1})} \cdot 2^{-poly(\varepsilon^{-1})}$.
First, notice that if we condition on $S \subseteq N_1$ then \Cref{lem:weak-coreset-for-1-mean} gives that $\mu(S)$ is a $(1+\varepsilon)$-approximation for $1$-mean on $N_1$ with constant probability. Thus, we are left to prove that $S \subseteq N_1$ with probability $k^{-O(\varepsilon^{-1})} \cdot 2^{-poly(\varepsilon^{-1})}$. 
We have that the $P_{p\sim \cost{p}{\calT}}[p \in N_1 \cond p \in N] \geq \varepsilon^2 / k$, however the costs w.r.t. $\calT$ of points in $N$ varies of at most a factor $poly(\varepsilon^{-1})$, thus $P_{p\sim Unif}[p \in N_1 \cond p \in N] \geq poly(\varepsilon)/k$. The probability of $S \subseteq N_1$ is thus $(poly(\varepsilon) /k)^{|S|} = k^{-O(\varepsilon^{-1})} \cdot 2^{-poly(\varepsilon^{-1})}$.
In this case, we set $\alpha_0 \leftarrow |N_1| / |\optcl_1|$ because $\mu(S)$ approximates the mean of the entire set $N_1$.

Case $(ii)$. Here we give up on estimating the mean of $N_1$ and set $\alpha_0 \leftarrow 0$. 
The point $x \in N_1$ such that $\calT[x] \not\in \calQ$ can be assigned to $\calT[x]$ incurring a combined cost of $\varepsilon^2 \opt / k$.
We partition the remaining points in $N_1$ into $Z_1 \cup \dots Z_s$ where each point $x$ is placed in $Z_i$ if $\calT[x] = q_i$. 
Now, we collapse the points in $Z_i$ so as to coincide with $q_i$ and show that this does not worsen our approximation factor.
In terms of coefficients $(\alpha_i)_i$, this translates into the updates $\alpha_i \leftarrow \alpha_i + |Z_i| / |\optcl_i|$ for each $i = 1 \dots s$. 

Indeed, using \Cref{obs:as-if-concentrated-obs} we can move all points in $Z_i$ to $q_i$ incurring an additive combined cost of $\varepsilon \opt / k$ and a multiplicative cost of $1+O(\varepsilon)$.

\paragraph{Points in $F_1$.} Points in $F_1$ are very far from $q_1$ and thus far from $o_1$, hence even if their contribution to $\cost{\optcl_1}{o_1}$ might be large, we have $\cost{F_1}{o_1} = (1 \pm O(\varepsilon)) \cost{F_1}{o'}$ for all $o'$ in a ball of radius $\rho_1 / \varepsilon$ centered in $o_1$, thanks to \Cref{lem:gradient-argument}.

Let $H$ be the set of points that have not been assigned to centers in $\calC \setminus Out$. In particular, $H = W \cup N_1$ if points in $N_1$ satisfy case $(i)$ and $H = W \cup Z_1 \dots Z_s$ if points in $N_1$ satisfy case $(ii)$.
We consider two cases. 

If $||\mu(H) - q_1 || \leq \rho / \varepsilon$, then $||\mu(H) - o_1 || \leq \rho (1 + \varepsilon + 1/ \varepsilon)$ because $q_1 \in \core{o_1}$. Since for each $f \in F_1$ we have
$||f - o_1|| \geq ||f - q_1|| - (1 + \varepsilon) \rho \geq \Omega(\rho / \varepsilon^3)$ then $\cost{f}{o'} = (1\pm O(\varepsilon)) \cost{f}{o_1}$ for each $o'$ in a ball of radius $O(\rho / \varepsilon)$ centered in $o_1$, and so in particular for $o' = \mu(H)$. Thus in this case we can simply disregard all points in $F_1$ and computing $\ho_1$ according to the $(\alpha_i)_i$ defined above yields a perfect approximation of $o_i$.

Else, if $||\mu(H) - q_1 || > \rho / \varepsilon$, a similar argument applies to show that $\cost{ H}{o'} = (1\pm \varepsilon) \cost{H}{o}$ for each $o'$ in ball of radius $O(\rho)$ centered in $o_1$. Indeed, we can rewrite $\cost{H}{o'}$ as $|H| \cdot \cost{\mu(H)}{o'} + \cost{\mu(H)}{H}$. If $||\mu(H) - q_1 || < \rho / \varepsilon$ the first term varies of at most a factor $(1+\varepsilon)$ and the second term is constant. Thus in this case $\ho_1 = q_1$ is a perfect approximation of $o_1$ and we simply set $\alpha_1 = 1$ and $\alpha_j = 0$ for $j\neq 1$. In other words, here $\mu(N_1 \cup H)$ is too far from $q_1$ (and thus $o_1$) to significantlyt influence the position of $\ho_1$ and the same holds for any point in $F_1$. This works, of course, because we assumed $q_1 \in \core{o_1}$.
\end{proof}

\paragraph{Discussing the limitations on the coefficients values.}
The proof above would work smoothly if we were allowed to set $\alpha_i$ to exactly the values discussed above, representing the fractions of points from $\optcl_i$ captured by different $q_i$s. However, to make the algorithm efficient we limit ourselves to values in $\set{1, (1-\varepsilon) \dots (1-\varepsilon)^{\lceil \log_{1-\varepsilon}(\varepsilon^{7})} \rceil} \cup \{0\}$. \Cref{lem:alpha-are-off-its-ok} shows that as long as the values of $(\alpha_i)_i$ estimate the frequencies described above up to a factor $1\pm O(\varepsilon)$ then the approximation error is within a multiplicative factor $1\pm O(\varepsilon)$.

We are left to take care of the case in which $\alpha_i$ is set to a value $< \varepsilon^{7}$. We set $\alpha_i$ when dealing with points in $C_1 \cup N_1$ and for each $x \in C_1 \cup N_1$ we have, for each $o' \in B(q_1, (1+\varepsilon)\rho)$, $\cost{x}{o'} \leq 2 \cost{q_1}{o'} + 2\cost{x}{q_1} = O(\rho_1 \varepsilon^{-6})$. Thus, if we simply set $\alpha_i \leftarrow 0$ whenever we have $\alpha_i < \varepsilon^{7}$ then the combined cost of points in $\optcl_1$ with respect to $o'$ varies by $\varepsilon^{7}|\optcl_1| \cdot \rho_1 \varepsilon^{-6} = O(\varepsilon) \opt_1$. Effectively, ignoring these points does not significantly impact the cost. hence solving $1$-mean ignoring these points finds a $(1+O(\varepsilon))$-approximate solution to the original problem. 

\section*{Additional Experimental Evaluation}
In this section we report additional experiments which presentation did not fit in the main body. In particular, we run experiments on the dataset KDD-PHY and for $k=10, 50$.

In \Cref{fig:appendix-vanilla-greedy} we compare \algogreedymsls{} with \algomsls{}. To perform our experiment, we initialize $k=25$ centers using \algokmeans{} and then run $50$ iterations of local search for both algorithms, for $p \in\{2, 3\}$ swaps. We repeat each experiment $5$ times.
For ease of comparison, we repeat the plot for the KDD-BIO and RNA datasets that we present in the main body of the paper.
Due to the higher running of the \algomsls{} we perform this experiments on 1\% uniform sample of each of our datasets. 
We find out that the performance of the two algorithms is comparable on all our instances, while they both perform roughly 15\%-27\% better than $k$-means++ at convergence.

\begin{figure}[] 
    \centering
     \begin{subfigure}
         \centering
        \includegraphics[width=0.47\textwidth]{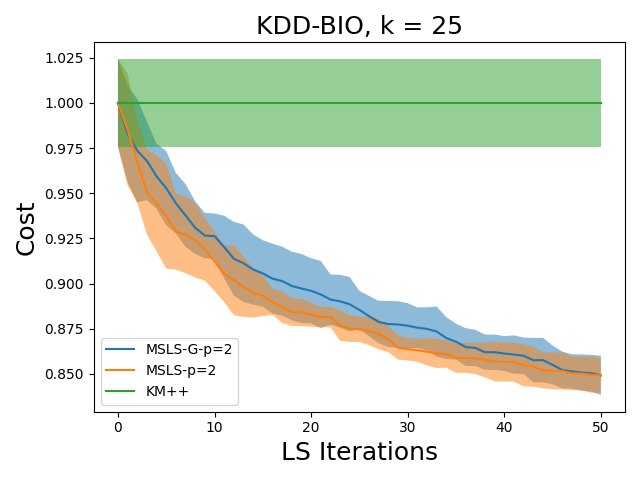}
    \end{subfigure}
         \begin{subfigure}
         \centering
        \includegraphics[width=0.47\textwidth]{./images/vvsg-3-swaps-KDD-BIO.png}
    \end{subfigure}
    \begin{subfigure}
         \centering
        \includegraphics[width=0.47\textwidth]{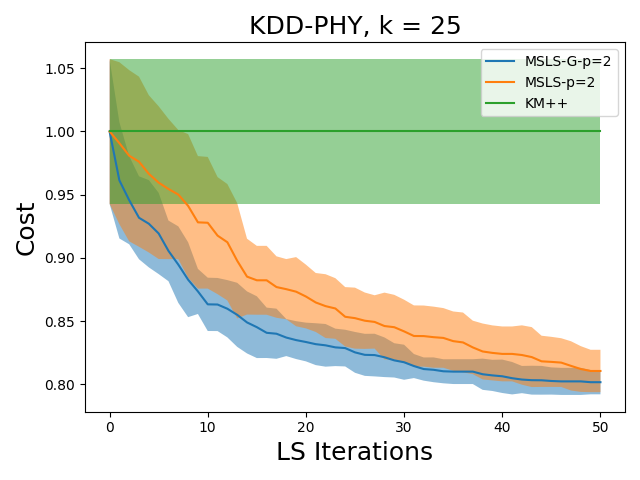}
        \end{subfigure}
     \begin{subfigure}
         \centering
        \includegraphics[width=0.47\textwidth]{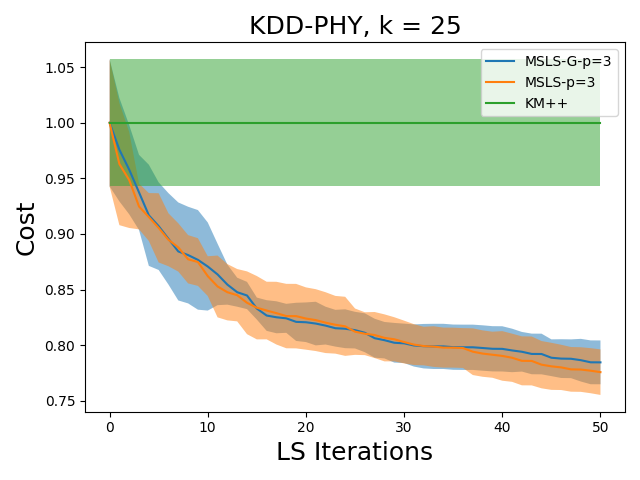}
    \end{subfigure}
     \begin{subfigure}
         \centering
        \includegraphics[width=0.47\textwidth]{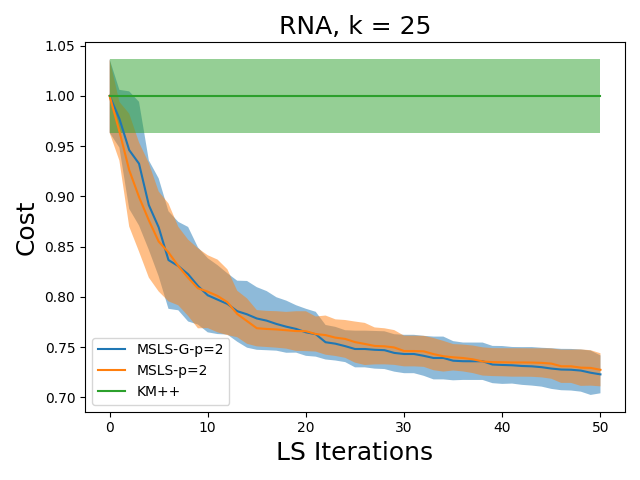}
    \end{subfigure}
     \begin{subfigure}
         \centering
        \includegraphics[width=0.47\textwidth]{./images/vvsg-3-swaps-RNA.png}
    \end{subfigure}
    \caption{Comparison between \algomsls{} and \algogreedymsls{}, for $p =2$ (left column) and $p=3$ (right column), for $k=25$, on the datasets KDD-BIO (first row), KDD-PHY (second row) and RNA (third row). The $y$ axis shows the mean solution cost, over the 5 repetitions of the experiment, divided by the means solution cost of \algokmeans{}.}
    \label{fig:appendix-vanilla-greedy}
\end{figure}

In \Cref{fig:appendix-nolloyd} we run \algokmeans{} followed by $50$ iterations of \algogreedymsls{} with $p = 1, 4, 7, 10$ and $k=10, 25, 50$ (expcluding the degenerate case $p=k=10$) and plot the relative cost w.r.t. \algokmeans{} at each iteration. The results for $k=25$ on KDD-BIO and RNA can be found in \Cref{fig:main-experiment}. We repeat each experiment $5$ times. Our experiment shows that, after $50$ iterations \algogreedymsls{} for $p = 4, 7, 10$ achieves improvements of roughly $5-10\%$ compared to \algogreedymslsarg{1} and of the order of $20\%-40\%$ compared to \algokmeans. These improvements are more prominent for $k=25, 50$. We also report the time per iteration that each algorithm takes. For comparison, we report the running time of a single iteration of Lloyd's next to the dataset's name. 
Notice that the experiment on RNA for $k=50$ is performed on a $10\%$ uniform sample of the original dataset, due to the high running time.

\begin{figure}[] 
    \centering
     \begin{subfigure}
         \centering
        \includegraphics[width=0.47\textwidth]{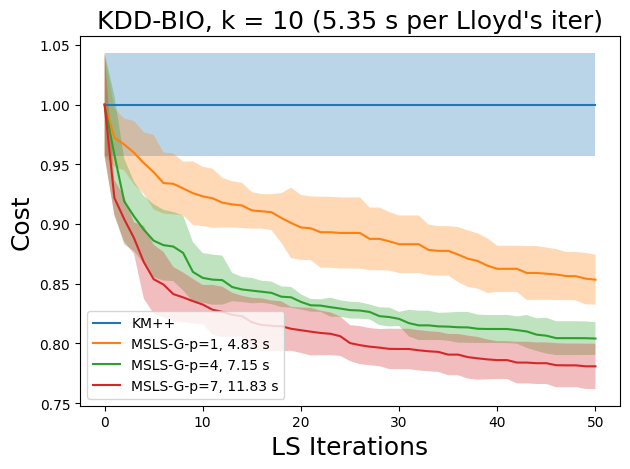}
    \end{subfigure}
    \begin{subfigure}
         \centering
        \includegraphics[width=0.47\textwidth]{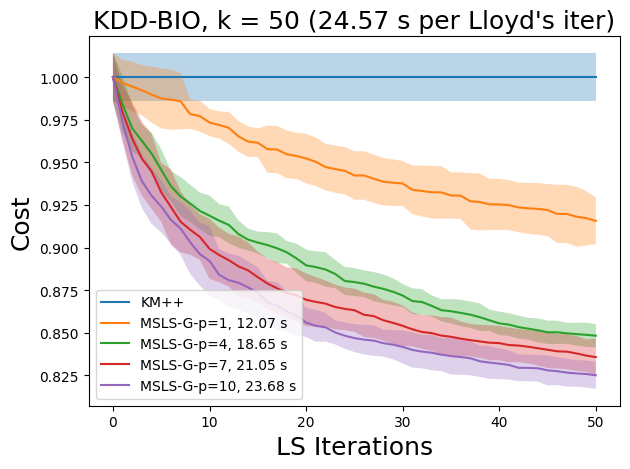}
        \end{subfigure}
             \begin{subfigure}
         \centering
        \includegraphics[width=0.47\textwidth]{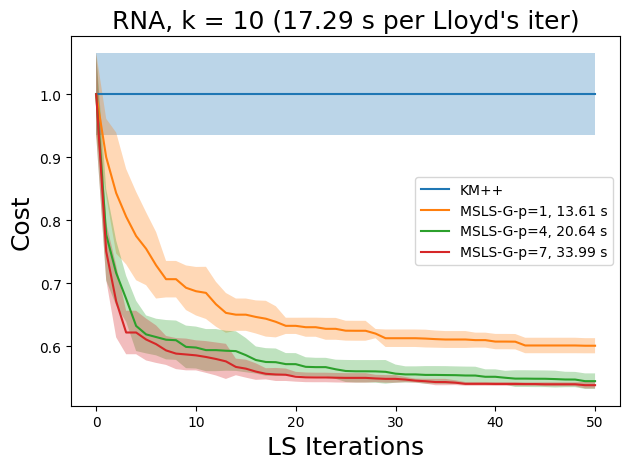}
    \end{subfigure}
         \begin{subfigure}
         \centering
        \includegraphics[width=0.47\textwidth]{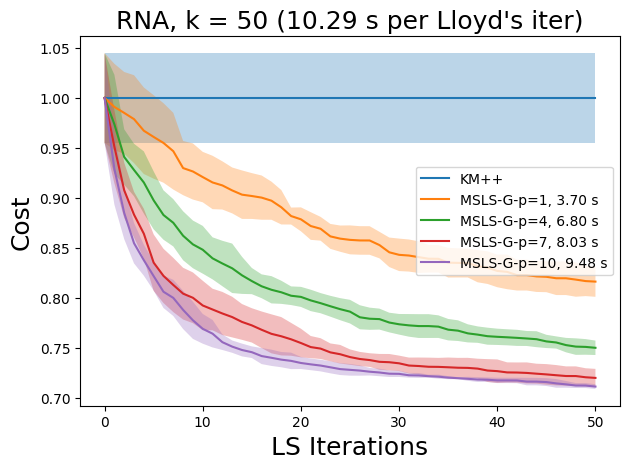}
    \end{subfigure}
     \begin{subfigure}
         \centering
        \includegraphics[width=0.47\textwidth]{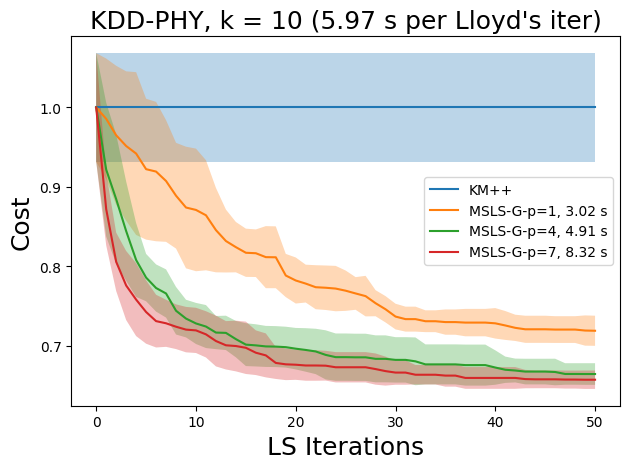}
    \end{subfigure}
     \begin{subfigure}
         \centering
        \includegraphics[width=0.47\textwidth]{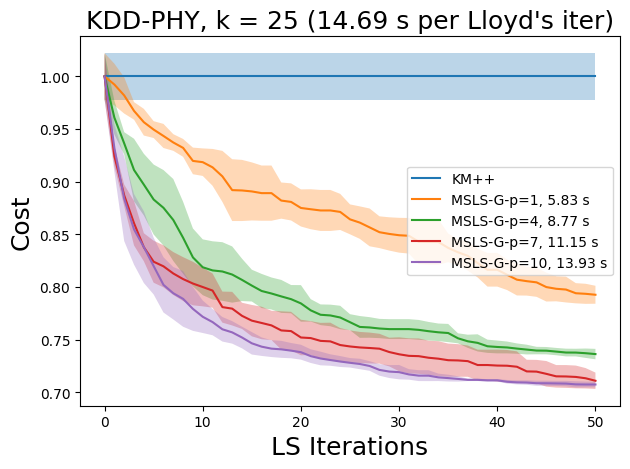}
    \end{subfigure}
     \begin{subfigure}
         \centering
        \includegraphics[width=0.47\textwidth]{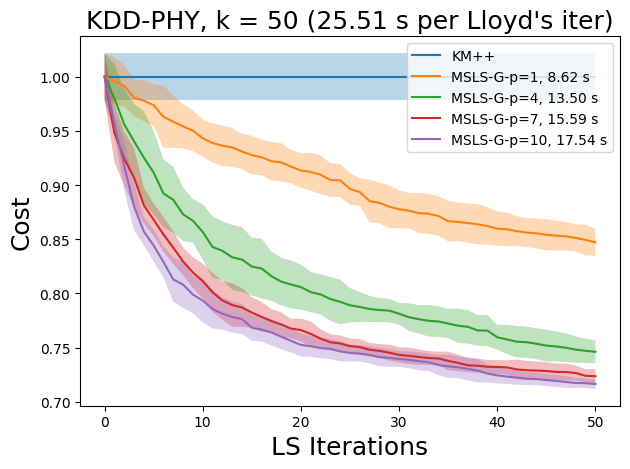}
    \end{subfigure}
    \caption{We compare the cost of \algogreedymsls{}, for $p\in\{1,4, 7, 10\}$, divided by the mean cost of \algokmeans{} at each LS  step, for $k\in\{10, 25, 50\}$, excluding the degenerate case $p=k=10$. The legend reports also the running time of \algogreedymsls{} per LS step (in seconds). The experiments were run on all datasets: KDD-BIO, RNA and KDD-PHY, excluding the case of $k=25$ for KDD-BIO and RNA which are reported in the main body of the paper.}
    \label{fig:appendix-nolloyd}
\end{figure}

In  \Cref{fig:appendix-lloyd}, we use \algokmeans{} and \algogreedymsls{} as a seeding algorithm for Lloyd's and measure how much of the performance improvement measured is retained after running Lloyd's. 
First, we initialize our centers using \algokmeans{} and the run $15$ iterations of \algogreedymsls{} for $p=1, 4, 7$.  We measure the cost achieved by running $10$ iterations of Lloyd's starting from the solutions found by \algogreedymsls{} as well as \algokmeans{}. We run experiments for $k=10, 25, 50$ and we repeat each experiment $5$ times. We observe that for $k=25, 50$ \algogreedymsls{} for $p>1$ performs at least as good as \algossls{} from \cite{silvio-original} and in some cases maintains non-trivial improvements. These improvements are not noticeable for $k=10$; however, given how Lloyd's behave for $k=10$ we conjecture that $k=10$ might be an ``unnatural'' number of clusters for our datasets.

\begin{figure}[] 
    \centering
     \begin{subfigure}
         \centering
        \includegraphics[width=0.47\textwidth]{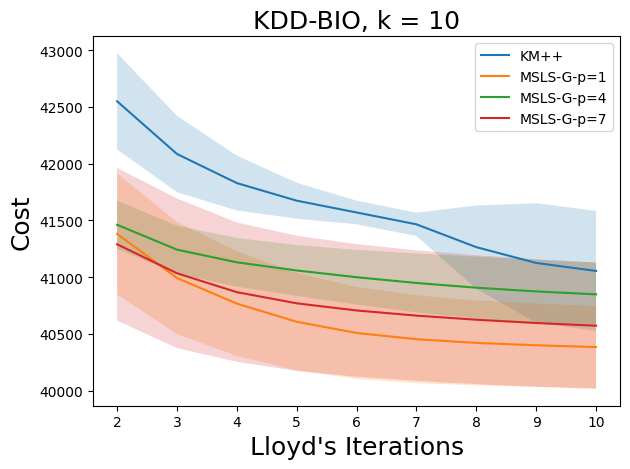}
    \end{subfigure}
         \begin{subfigure}
         \centering
        \includegraphics[width=0.47\textwidth]{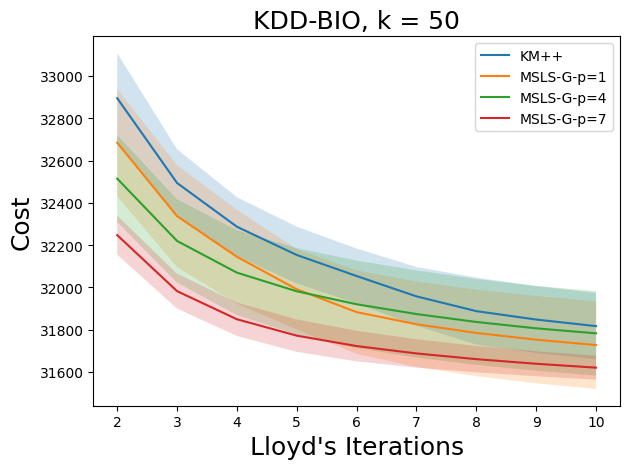}
    \end{subfigure}
    \begin{subfigure}
         \centering
        \includegraphics[width=0.47\textwidth]{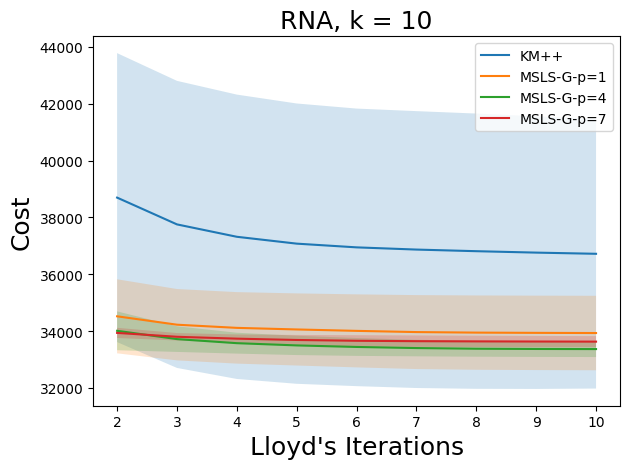}
        \end{subfigure}
     \begin{subfigure}
         \centering
        \includegraphics[width=0.47\textwidth]{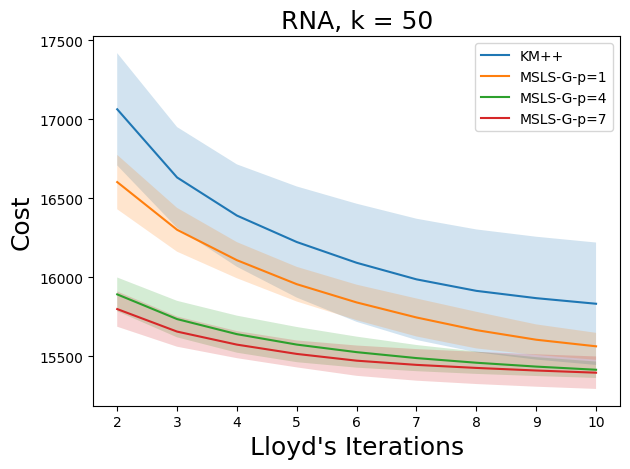}
    \end{subfigure}
     \begin{subfigure}
         \centering
        \includegraphics[width=0.47\textwidth]{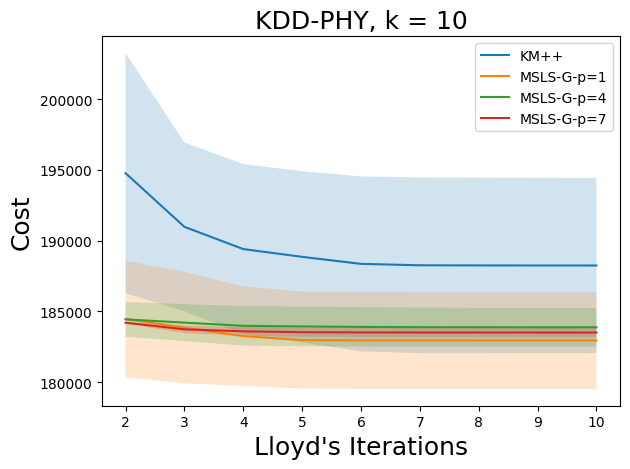}
    \end{subfigure}
     \begin{subfigure}
         \centering
        \includegraphics[width=0.47\textwidth]{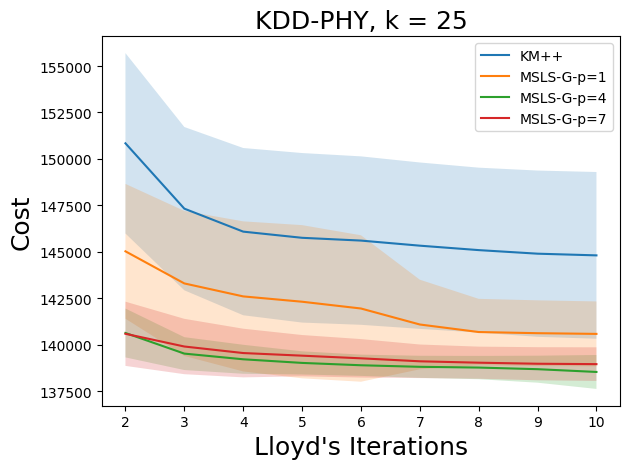}
    \end{subfigure}
     \begin{subfigure}
         \centering
        \includegraphics[width=0.47\textwidth]{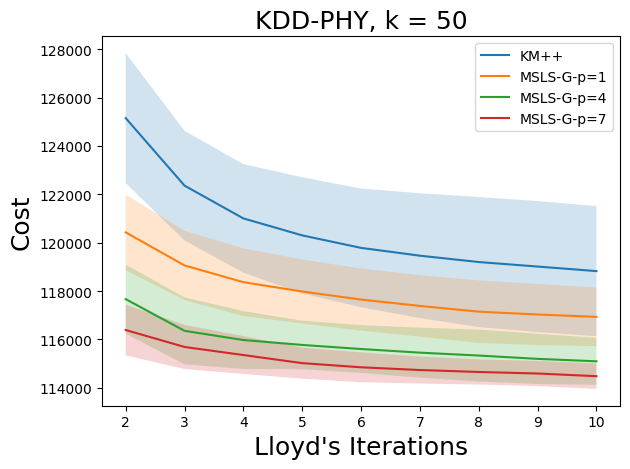}
    \end{subfigure}
    \caption{We compare the cost after each of the $10$ iterations of Lloyd with seeding from \algogreedymsls{}, for $p\in\{1,4, 7, 10\}$ and $15$ local search steps and \algokmeans{}, for $k\in \{10, 25, 50\}$. We excluded the degenerate case $p=k=10$, and the experiments reported in the main body of the paper.}
    \label{fig:appendix-lloyd}
\end{figure}

\end{document}